\makeatletter
\def\UTFActivateAnd#1{\bgroup\def\UTFviii@defined##1{\egroup#1##1}}
\UTFActivateAnd\def定{\UTFActivateAnd\def}定令{\UTFActivateAnd\let}

\documentclass{IEEEtran}

\usepackage{mathtools,amssymb,bm,amsthm,booktabs}
	\allowdisplaybreaks
	令Γ\Gamma 令Δ\Delta 令Θ\Theta 令Λ\Lambda 令Ξ\Xi 令Π\Pi 令Σ\Sigma
	令Υ\Upsilon 令Φ\Phi 令Ψ\Psi 令Ω\Omega 令α\alpha 令β\beta 令γ\gamma
	令δ\delta 令ε\varepsilon 令ζ\zeta 令η\eta 令θ\theta 令ι\iota 令κ\kappa
	令λ\lambda 令μ\mu 令ν\nu 令ξ\xi 令π\pi 令ρ\rho 令ς\varsigma 令σ\sigma
	令τ\tau 令υ\upsilon 令φ\varphi 令χ\chi 令ψ\psi 令ω\omega 令ϑ\vartheta
	令ϕ\phi 令ϖ\varpi 令Ϝ\Digamma 令ϝ\digamma 令ϰ\varkappa 令ϱ\varrho
	令ϴ\varTheta 令ϵ\epsilona
	定𝛂{\bm\alpha} 定𝛃{\bm\beta} 定𝛄{\bm\gamma} 定𝛅{\bm\delta}
	定𝛆{\bm\varepsilon} 定𝛇{\bm\zeta} 定𝛈{\bm\eta} 定𝛉{\bm\theta}
	定𝛊{\bm\iota} 定𝛋{\bm\kappa} 定𝛌{\bm\lambda} 定𝛍{\bm\mu}
	定𝛎{\bm\nu} 定𝛏{\bm\xi} 定𝛐{\bm\omicron} 定𝛑{\bm\pi} 定𝛒{\bm\rho}
	定𝛓{\bm\varsigma} 定𝛔{\bm\sigma} 定𝛕{\bm\tau} 定𝛖{\bm\upsilon}
	定𝛗{\bm\varphi} 定𝛘{\bm\chi} 定𝛙{\bm\psi} 定𝛚{\bm\omega}
	定𝛛{\bm\partial} 定𝛜{\bm\epsilon} 定𝛡{\bm\varpi} 定𝛝{\bm\vartheta}
	定𝛞{\bm\varkappa} 定𝛟{\bm\phi} 定𝛠{\bm\varrho} 定𝐛{\bm b} 定𝐜{\bm
	c} 定𝐯{\bm v} 定𝐱{\bm x} 定𝐲{\bm y} 定𝟎{\bm0} 定𝟏{\bm1}
	定𝔼{\mathbb E} 定𝔽{\mathbb F} 定ℕ{\mathbb N} 定ℙ{\mathbb P}
	定ℝ{\mathbb R} 定ℤ{\mathbb Z} 定𝒞{\mathcal C} 定𝒟{\mathcal D}
	定ℰ{\mathcal E} 定ℱ{\mathcal F} 定𝒢{\mathcal G} 定𝒮{\mathcal S}
	定𝒯{\mathcal T} 定𝒴{\mathcal Y} 定𝒵{\mathcal Z}
	令ℓ\ell 令∂\partial 令∇\nabla
	令√\sqrt 令⌈\lceil 令⌉\rceil 令⌊\lfloor 令⌋\rfloor 令⟨\langle 令⟩\rangle
	令｜\mid 定＼{\setminus\nolinebreak} 令：\colon
	令´\acute 令ˆ\hat 令˜\tilde 令¯\bar 令˘\breve 令˙\dot 令¨\ddot
	令°\mathring 令ˇ\check
	令∏\prod 令∑\sum 令∫\int 令⋀\bigwedge 令⋁\bigvee 令⋂\bigcap 令⋃\bigcup
	令⨁\bigoplus 令⨂\bigotimes
	令±\pm 令·\cdot 令×\times 令÷\frac 令†\dagger 令•\bullet 令∖\setminus
	令∘\circ 令∧\wedge 令∨\vee 令∩\cap 令∪\cup 令⊕\oplus 令⊗\otimes 令⊙\odot
	令⋆\star
	令¬\neg 令…\ldots 令∀\forall 令∁\complement 令∃\exists 令∞\infty 令⊤\top
	令⊥\bot 令⋯\cdots 令★\bigstar
	令←\gets 令↑\uparrow 令→\to 令↓\downarrow 令↔\leftrightarrow 令↖\nwarrow
	令↗\nearrow 令↘\searrow 令↙\swarrow 令↞\twoheadleftarrow
	令↠\twoheadrightarrow 令↤\mapsfrom 令↦\mapsto 令↩\hookleftarrow
	令↪\hookrightarrow 令↾\upharpoonright 令∈\in 令∉\notin 令∋\ni 令∼\sim
	令≅\cong 令≈\approx 定≔{\coloneqq} 令≕\eqqcolon 令≠\neq 令≡\equiv
	定≢{\not\equiv} 令≤\leqslant 令≥\geqslant 令≪\ll 令≫\gg 令⊆\subseteq
	令⊇\supseteq 令⋮\vdots 令⋰\adots 令⋱\ddots 令⟂\perp 令♭\flat 令♮\natural
	令♯\sharp 令⟵\longleftarrow 令⟶\longrightarrow 令⟼\longmapsto
	定⁰{^0} 定¹{^1} 定²{^2} 定³{^3} 定⁴{^4} 定⁵{^5} 定⁶{^6} 定⁷{^7} 定⁸{^8} 定⁹{^9}
	定₀{_0} 定₁{_1} 定₂{_2} 定₃{_3} 定₄{_4} 定₅{_5} 定₆{_6} 定₇{_7} 定₈{_8} 定₉{_9}
	令㏒\log 令㏑\ln
	定～{\kern1em} 定！{\kern-1em}
	定†#1†{\text{#1}}
	定¬#1¬{\overline{#1}}
	\@namedef[{\begin{equation*}}
	\@namedef]{\end{equation*}}
	\DeclareMathOperator\TC{Tri\_Convexify}
	\DeclareMathOperator\LI{Linear\_Interp}
	\DeclareMathOperator\argmax{arg\,max}
	\def\BMS{\mathcal{B\!M\!S}\@ifstar{\!_\diamondsuit}{}}
	\DeclareMathOperator\BSC{BSC}
	\DeclareMathOperator\BEC{BEC}
	\def\parall{\mathbin{\ooalign{%
		$\odot$\cr
		\hss\color{white}$\bullet$\hss\cr
		\hss\rotatebox[origin=cc]{120}{%
			$\vphantom\odot$\clap{\raisebox{.023em}{$\star$}}%
		}\hss\cr
	}}}
	\def\serial{\mathbin{\ooalign{%
		$\phantom\odot$\cr
		\hss\rule[-.05em]{.6em}{.6em}\hss\cr
		\hss\color{white}\rule[-.02em]{.54em}{.54em}\hss\cr
		\hss\rotatebox[origin=cc]{240}{%
			$\vphantom\odot$\clap{\raisebox{.023em}{$\star$}}%
		}\hss\cr
	}}}
	\def\para{{\ooalign{%
		$\scriptstyle\odot$\cr
		\hss\color{white}$\scriptstyle\bullet$\hss\cr
	}}}
	\def\seri{{\ooalign{%
		$\scriptstyle\phantom\odot$\cr
		\hss\rule[-.05em]{.45em}{.45em}\hss\cr
		\hss\color{white}\rule[-.02em]{.39em}{.39em}\hss\cr
	}}}
	令Ｓ\serial
	令Ｐ\parall
	令ｓ\seri
	令ｐ\para

\usepackage{xcolor}
	定色#1♯#2 {\definecolor{#1}{HTML}{#2}}
	色PMS2767♯182B49   色PMS3015♯00629B   色PMS1245♯C69214   色PMS116♯FFCD00     
	色PMS3115♯00C6D7   色PMS7490♯6E963B   色PMS3945♯F3E500   色PMS144♯FC8900     
	色Black♯000000     色CoolGray9♯747678 色PMS401♯B6B1A9    色Metallic871♯84754E

\usepackage{tikz,pgfplotstable}
	\tikzset{every picture/.style={line cap=round,line join=round}}
	\pgfplotsset{compat/show suggested version=false,compat=1.17}

\usepackage{csquotes}

\usepackage[style=alphabetic,maxbibnames=99,
	maxalphanames=4,minalphanames=3]{biblatex}
\usepackage[hyphenbreaks]{breakurl}

\usepackage[unicode,pdfusetitle]{hyperref}
	\hypersetup{
		pdfsubject={Information Theory (cs.IT); 94B65},
		pdfkeywords={polar code, scaling exponent},
		colorlinks,allcolors=PMS2767,
	}

\usepackage{cleveref}
	
	定名#1?#2:#3 {\crefname{#1}{#2}{#3}}
		名section?Section:Sections
		名appendix?Appendix:Appendices
		名figure?Figure:Figure
		名table?Table:Tables
	定理#1?#2:#3 {\newtheorem{#1}[allams]{#2}名#1?#2:#3 }
		理thm?Theorem:Theorems
		理lem?Lemma:Lemmas
		理pro?Proposition:Propositions
		理cor?Corollary:Corollaries
	\theoremstyle{definition}
		理dfn?Definition:Definitions
		理exa?Example:Examples
	\theoremstyle{remark}
		理rem?Remark:Remarks
		理axi?Assumption:Assumptions
		理cla?Claim:Claims
	定式#1?#2:#3 {名#1?#2:#3 \creflabelformat{#1}{\textup{(##2##1##3)}}}
		式equ?equation:equations
		式ine?inequality:inequalities
		式for?formula:formulas
		式rat?ratio:ratios
		式cri?criterion:criteria
		式sup?supremum:suprema
		式gri?grid:grids
	定標#1:#2?{\label@in@display@optarg[#1]{#1:#2}}\AtBeginDocument{
		\def\label@in@display#1{\incr@eqnum\tag{\theequation}標#1?}}

\begin{document}

\title{
                    Sub-4.7 Scaling Exponent of Polar Codes
}
\author{
                                  Hsin-Po Wang
                                      and
                                 Ting-Chun Lin
                                      and
                                Alexander Vardy
                                      and
                                  Ryan Gabrys%
\thanks{
                              The authors are with
                  University of California San Diego, CA, USA.
                                Lin is also with
             Hon Hai (Foxconn) Research Institute, Taipei, Taiwan.
                           This work was supported by
                    NSF grants CCF-1764104 and CCF-2107346.
             Emails: \{hsw001, til022, avardy, rgabrys\} @ucsd.edu
}}
\maketitle

\advance\baselineskip0pt plus.25pt minus.125pt

\begin{abstract}\boldmath
	Polar code visibly approaches channel capacity in practice and is thereby
	a constituent code of the 5G standard.  Compared to low-density parity-check
	code, however, the performance of short-length polar code has rooms for
	improvement that could hinder its adoption by a wider class of applications.
	As part of the program that addresses the performance issue at short length,
	it is crucial to understand how fast binary memoryless symmetric channels
	polarize.  A number, called scaling exponent, was defined to measure the
	speed of polarization and several estimates of the scaling exponent were
	given in literature.  As of 2022, the tightest overestimate is $4.714$ made
	by Mondelli, Hassani, and Urbanke in 2015.  We lower the overestimate to
	$4.63$.
\end{abstract}

\section{Introduction}

	\IEEEPARstart P{olar code} was proved to be capacity achieving over any
	binary memoryless symmetric (BMS) channel \cite{Arikan09}.  Polar code also
	shows great potential in practice and it was selected as part of the 5G
	standard for wireless communication.  That being the case, polar coding for
	short block length has room for improvement when compared to low-density
	parity-check code, the other code in the 5G standard.  Improving
	short-length polar code further can pave the way for applications such as
	Internet of Things, as some devices can only afford easily-decodable code
	and others must reply very promptly.

	Now that improving the performance of polar code at finite block length is
	on the agenda, we first need to know how much we can say about the
	unmodified code.  There are two regimes that were considered in literature.
	In the error exponent regime, the code rate is fixed and the asymptote of
	the error probability is evaluated.  For polar code, it was shown that the
	block error probability scales as $\exp(-√N)$, where $N$ is the block
	length.  For variations of polar code that use different matrices as the
	polarizing kernel, the asymptote of error can also be computed and is about
	$\exp(-N^β)$.  Here, $β > 0$ is a number completely determined by the
	Hamming distances among the vector subspaces spanned by the rows of the
	kernel matrix.  Long story short, predicting the behavior of error
	probability at a fixed code rate is straightforward.  See
	\cite{AT09,KSU10,HMTU13,MT14} and those that cite them for more on this
	topic.

\begin{table}
	\caption{
		Two key ways to synthesize channels. \\
		SC stands for sequential cancellation decoder.
	}\label{tab:compare}
	\def\arraystretch{1.2}
	\tabcolsep4pt
	\centering\begin{tabular}{cc}
		\toprule
		$WＳW$                           & $WＰW$                             \\
		\midrule
		serial combination               & parallel combination              \\
		convolution at check node        & convolution at variable node      \\
		guess $X₁-X₂$ given $Y₁, Y₂$     & guess $X₂$ given $Y₁, Y₂, X₁-X₂$  \\
		decoded earlier in SC            & decoded later in SC               \\
		named $W'$ or $W^-$ or $W^{(1)}$ & named $W''$ or $W^+$ or $W^{(2)}$ \\
		more noisy than $W$              & more reliable than $W$            \\
		still BSC if $W$ is              & not BSC if $Z(W) ≠ 0,1$           \\
		$1 - Z(WＳW) ≥ (1 - Z(W))²$      & $Z(WＰW) = Z(W)²$                  \\
		\bottomrule
	\end{tabular}
\end{table}

	In the scaling exponent regime, the second approach that characterizes the
	performance of polar code, the error probability is fixed and the asymptote
	of the code rate is evaluated.  It is observed that the \emph{gap to
	capacity}, which is the difference between the channel capacity and code
	rate, scales as $N^{-1/μ}$.  Called the \emph{scaling exponent}, this number
	$μ$ is difficult to pinpoint exactly.  Here is a list of progresses made
	before.
	It was shown in \cite{HAU10} that $0.2786 ≥ 1/μ ≥ 0.2669$
	over binary erasure channels (BECs).
	It was shown in \cite{KMTU10} that $μ ≈ 3.626$ over BECs.
	It was shown in \cite{GHU12} that $3.553 ≤ μ$ over BMS channels.
	It was shown in \cite{HAU14} that $3.579 ≤ μ ≤ 6$ over BMS channels.
	It was shown in \cite{GB14} that  $μ ≤ 5.702$ over BMS channels.
	Mondelli, Hassani, and Urbanke
	showed in \cite{MHU16} that $μ ≤ 4.714$ over BMS channels.  The last record
	stood for seven years\footnote{ The preprint was first released in January
	2015 at \url{https://arxiv.org/abs/1501.02444}} and is the one we intend to
	improve upon.
	
	Scaling exponent's definition generalizes to other scenarios.
	To name a few: 
	Over additive white Gaussian noise channels, $μ ≤ 4.714$ \cite{FT17}.
	Over non-stationary BECs, $μ ≤ 7.34$;
	over non-stationary BMS channels, $μ ≤ 8.54$ \cite{Mahdavifar20}.
	Over (hereafter stationary) BECs, permuting the rows of the Kronecker powers
	of Arıkan's kernel $[¹₁{}⁰₁]$ improves the scaling from $μ ≈ 3.627$
	to $μ ≈ 3.479$ with little complexity overhead \cite{BFSTV17}.
	Using larger kernel matrices improves scaling exponents even further:
	over BECs,
	$μ ≈ 3.627$ for $2 × 2$ kernel,
	$μ ≈ 3.577$ for $8 × 8$ kernel \cite{FV14},
	$μ ≈ 3.346$ for $16 × 16$ kernel \cite{TT21},
	$μ ≈ 3.122$ for $32 × 32$ kernel,
	and $μ ≈ 2.87$ for $64 × 64$ kernel\footnote{The scaling exponent
	for the $64 × 64$ kernel involves Monte Carlo method.} \cite{YFV19}.
	(See \cite{Trofimiuk21s} for sizes between $9 × 9$ and $31 × 31$.)
	In general, any nontrivial matrix kernel over any alphabet
	has a finite scaling exponent over any discrete memoryless channel
	\cite[Chapter~5]{Chilly}.
	Meanwhile, dynamic kerneling is also shown, conceptually,
	to be improving the scaling exponent; for instance,
	$μ ≈ 4.938$ decreases to $μ ≈ 4.183$ for $3 × 3$ kernels over BECs\footnote{
	We recalculate the exponents for dynamic kerneling using
	power iteration to even the baseline for comparison.} \cite{YB15}.
	Most challengingly, a series of works attempted to reach $μ ≈ 2$,
	the optimal scaling exponent, and succeeded.
	Pfister and Urbanke \cite{PU19} showed that $μ ≈ 2$ can be reached
	using Reed--Solomon kernels over $q$-ary erasure channels as $q → ∞$.
	Fazeli, Hassani, Mondelli, and Vardy \cite{FHMV21} showed that
	$μ ≈ 2$ can be reached using random linear kernel over BECs.
	Guruswami, Riazanov, and Ye \cite{GRY22} showed that $μ ≈ 2$ can be
	reached using dynamic random linear kernels over BMS channels;
	plus the code construction is of polynomial complexity.
	Wang and Duursma \cite{Hypotenuse} showed that $μ ≈ 2$ can be reached
	using dynamic random linear kernels over discrete memoryless channels.

	A good scaling exponent over BMS channels has several boarder impacts.
	One: One can now describe the trade-off between gap to capacity anderror
	probability;  this is called the moderate deviation regime
	\cite[Section~2.6]{Chilly}.
	Two: For simplified decoders, the scaling exponent dictates how much
	soft-decision can be pruned away and controls the complexity
	\cite{LoglogTime}.
	Three: For parallelized decoders, the scaling exponent dictates how much
	work still needs to be processed in serial and controls the latency
	\cite{HMFVCG21}.
	Four: Polar code achieves the asymmetric capacity of any binary-input
	channel using the technique introduced in \cite{HY13};  the corresponding
	scaling exponent assumes the same estimate as BMS does
	\cite[Chapter~3]{Chilly}.  In fact, polar code achieves the same scaling
	exponent over discrete memoryless channels with constructions given in
	\cite{RWLP22}.
	Five: For lossless \cite{Arikan10,CK10} and lossy \cite{KU10} compression
	via polar coding, scaling exponent can be defined similarly and assumes the
	same bound \cite[Chapter~3]{Chilly}.
	Six: For multiple access channel, rate-splitting helps avoid time-sharing
	and achieve the same scaling exponent \cite{CY18}; for distributed lossless
	compression, a similar technique applies \cite[Chapter~8]{Chilly}.
	Seven: Over wiretap channels, polar code achieves the secrecy capacity but
	consumes secrete keys shared between Alice and Bob;  the scaling exponent
	gives prediction on the length of the secrete key \cite{WU16,GB17}.  Eight:
	For coded computation, scaling behavior is related to not only the code rate
	but also the waiting time \cite{FM22}.

	The goal of this paper is to improve $μ ≤ 4.714$ to $μ ≤ 4.63$.  The key
	idea is that a parallel combining followed by a serial combining makes a
	channel ``less BSC'' and hence some inequalities can be strengthened.  On
	the execution side, we remix a handful of techniques that are versatile and
	flexible:  We compute numerical convex envelopes to force functions become
	convex to apply Jensen's inequality; we use interval arithmetic library to
	obtain mathematically rigorous bounds to compensate coarse sampling; we use
	power iteration with a finite state automata to ``remember'' recent history.

	This paper is organized as follows.
	\Cref{sec:pre} reviews notations and preliminary results.
	\Cref{sec:old} reiterates the old proof of $μ ≤ 4.714$; we did not add
		anything new; the intention is to provide a baseline for comparison.
	\Cref{sec:tri} introduces tri-variate channel transformation $(U Ｐ V) Ｓ W$
		and the corresponding Bhattacharyya parameter inequalities.
	\Cref{sec:dfa} demonstrates how to use power iterations with memory
		to utilize the new Bhattacharyya parameter inequalities.
	\Cref{sec:wrap} wraps up the proof of the new result $μ ≤ 4.63$.

\section{Preliminary}\label{sec:pre}

\subsection{Binary memoryless symmetric channels}

\begin{figure}\centering
	\begin{tikzpicture}
		\draw[xscale=4]
			(0,2)node(0){$0$}(1,2)node(a){$0$}
			(0,0)node(1){$1$}(1,0)node(b){$1$}
		;
		\draw
			(0)edge[->]node[auto]{$1-p$}(a)
			(0)edge[->]node[auto,',pos=0.2]{$p$}(b)
			(1)edge[->]node[auto,pos=0.2]{$p$}(a)
			(1)edge[->]node[auto,']{$1-p$}(b)
		;
	\end{tikzpicture}
	\caption{
		$\BSC(p)$, binary symmetric channel with crossover
		probability $p$.
	}\label{fig:BSC}
\end{figure}
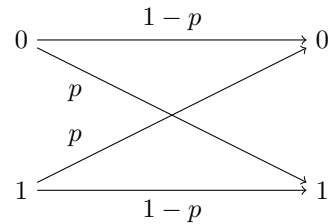

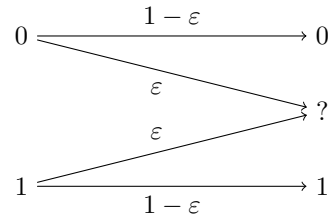
\begin{figure}\centering
	\begin{tikzpicture}
		\draw[xscale=4]
			(0,2)node(0){$0$}(1,2)node(a){$0$}
							 (1,1)node(e){$?$}
			(0,0)node(1){$1$}(1,0)node(b){$1$}
		;
		\draw
			(0)edge[->]node[auto]{$1-ε$}(a)
			(0)edge[->]node[auto,']{$ε$}(e)
			(1)edge[->]node[auto]{$ε$}(e)
			(1)edge[->]node[auto,']{$1-ε$}(b)
		;
	\end{tikzpicture}
	\caption{
		$\BEC(ε)$, binary erasure channel with erasure
		probability $ε$.
	}\label{fig:BEC}
\end{figure}

	A \emph{binary symmetric channel} (BSC) with crossover probability $p$ is a
	channel where a user feeds in a $0$ or a $1$ and it outputs what is fed with
	probability $1-p$ or flips the bit with probability $p$.  We denote it by
	$\BSC(p)$ and picture it in \cref{fig:BSC}.

	A \emph{binary erasure channel} (BEC) with erasure probability $ε$ is a
	channel where a user feeds in a $0$ or a $1$ and it outputs what is fed with
	probability $1-ε$ or outputs a question mark with probability $ε$.  We
	denote it by $\BEC(ε)$ and picture it in \cref{fig:BEC}.

	A \emph{binary memoryless symmetric} (BMS) generalizes BSC and BEC.  It is a
	channel where a user feeds in a $0$ or a $1$ and it outputs a symbol
	randomly selected from an alphabet set $𝒴$.  For a BMS channel $W$, the
	conditional probabilities of outputting $y ∈ 𝒴$ conditioning on inputs $0$
	and $1$ are denoted by $W(y|0)$ and $W(y|1)$, respectively.  A BMS channel
	is memoryless in the sense that repeated uses of this channel does not alter
	the conditional distribution.  A BMS channel $W$ is symmetric in the sense
	that for any output symbol $y ∈ 𝒴$, there is another symbol $¯y ∈ 𝒴$ such
	that $W(¯y|0) = W(y|1)$ and $W(¯y|1) = W(y|0)$.

\subsection{Channel equivalence and channel decomposition}

\begin{figure}\centering
	\begin{tikzpicture}
		\draw[xscale=2]
							 (1,3)node(a){$0$}(3,3)node(w){$0$}
			(0,2)node(0){$0$}(1,2)node(b){$1$}(3,2)node(x){$1$}
			(0,1)node(1){$1$}(1,1)node(c){$0$}(3,1)node(y){$?$}
							 (1,0)node(d){$1$}(3,0)node(z){$?$}
		;
		\draw
			(0)edge[->]node[auto,sloped]{$1-ε$}(a)
			(0)edge[->]node[auto,pos=0.3,sloped]{$ε$}(c)
			(1)edge[->]node[auto,',pos=0.3,sloped]{$1-ε$}(b)
			(1)edge[->]node[auto,',sloped]{$ε$}(d)
		;
		\draw[PMS3015]
			(a)edge[->]node[auto]{$1$}(w)
			(b)edge[->]node[auto]{$1$}(x)
		;
		\draw[PMS1245]
			(c)edge[->]node[auto]{$1/2$}(y)
			(c)edge[->]node[auto,pos=0.8]{$1/2$}(z)
			(d)edge[->]node[auto,pos=0.2]{$1/2$}(y)
			(d)edge[->]node[auto,']{$1/2$}(z)
		;
	\end{tikzpicture}	
	\caption{
		Illustration of $\BEC(ε)$ as undergoing ${\color{PMS3015}\BSC(0)}$ with
		frequency $1-ε$ and undergoing ${\color{PMS1245}\BSC(1/2)}$ with
		frequency $ε$.  Cf.\ \cite[Fig.~2.1]{LH06}.
	}\label{fig:BECasBSC}
\end{figure}
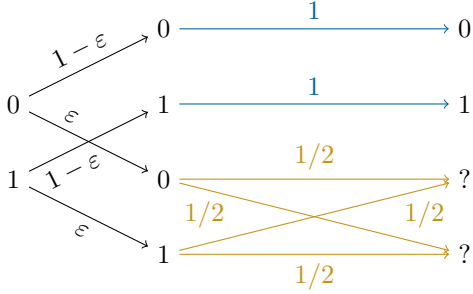

	Channels can have arbitrary output alphabets, but those that pose the same
	coding challenge are usually treated as the same.  An equivalence relation
	on the class of BMS channels is thus defined to identify and distinguish
	channels.

	We say that a BMS channel $W： \{0,1\} → 𝒵$ is a \emph{symbol aggregation}
	of another BMS channel  $V： \{0,1\} → 𝒴$ if there exists a map $π： 𝒴 → 𝒵$
	such that
	\begin{align*}
		V(y|0) : V(y|1) & = W(π(y)|0) : W(π(y)|1), \\
		∑_{z∈π^{-1}(z)} V(υ|0) + V(υ|1) & = W(z|0) + W(z|1)
	\end{align*}
	for all $y ∈ 𝒴$ and $z ∈ 𝒵$.  One sees that the purpose of $π$ is to
	identify symbols sharing the same likelihood ratio. Two BMS channels are
	said to be \emph{equivalent} if they share a common symbol aggregation.

	This equivalence relation on BMS channels extends to a partial ordering.  A
	BMS $W$ is said to be a \emph{degradation} of $V$ if $W$ can be obtained by
	post-processing the output of $V$.  (For instance, symbol aggregation counts
	as post-processing.)  It can be shown that $V$ and $W$ are equivalent iff
	$V$ is a degradation of $W$ and $W$ is a degradation of $V$.  For more on
	this viewpoint, see how to construct polar codes \cite{TV13}, how to deal
	with general alphabet \cite{GYB18}, how to describe input-degradation
	\cite{Nasser18}, and how output-degradation is used to achieve $μ = 2$
	within polynomial complexity \cite{GRY22}.
	
	Let $\BMS$ be the set of equivalence classes of BMS channels.  Let $\BMS*$
	be the set of equivalence classes excluding the noiseless channel ($W(y|0)
	W(y|1) = 0$ for all $y$) and the jammed channel ($W(y|0) = W(y|1)$ for all
	$y$).  What remain are the nontrivial channels where coding is meaningful.
	Later when Bhattacharyya parameter $Z$ is defined, one will see that $\BMS*$
	are channels with $Z(W) ∉ \{0,1\}$.

	Every BMS channel $W$ assumes a \emph{BSC-decomposition}
	\[W = ∑_j α_j \BSC(p_j),\]
	where $∑_j α_j = 1$ and $0 ≤ p_j ≤ 1/2$.  This notation means that $W$ can
	be simulated by (is equivalent to) the following procedure:
	\begin{itemize}
		\item select $\BSC(p_j)$ with probability $α_j$,
		\item reveal $p_j$, and
		\item feed the input into $\BSC(p_j)$ and reveal the BSC's output.
	\end{itemize}
	As an example, \cref{fig:BECasBSC} pictures the decomposition of $\BEC(ε)$
	into $(1-ε) \BSC(0) + ε \BSC(1/2)$.

	In general, the BSC-decomposition of a BMS channel $W： \{0,1\} → 𝒴$ can be
	obtained by the following procedure: First, aggregate all output symbols
	that share the same likelihood ratio.  Now that $W(y|0) : W(y|1)$ are all
	distinct for all $y ∈ 𝒴$, enumerate the output alphabet $𝒴 = \{ y₁, …,
	y_{|𝒴|} \}$, let $p_j ≤ 1/2$ be such that $1-p_j : p_j = W(y_j|0) :
	W(y_j|1)$ for all $y_j$ such that $W(y_j|0) ≥ W(y_j|1)$, and then let $α_j$
	be $W(y_j|0) + W(y_j|1)$.  For more on this topic, see \cite{GR20} and
	Modern Coding Theory \cite[Chapter~4]{RU08}.

\subsection{Bhattacharyya parameter}

	The \emph{Bhattacharyya parameter} of a BMS channel $W$ is denoted by
	$Z(W)$.  It is defined to be $Z(\BSC(p)) ≔ 2√{p¯p}$ for BSCs, where $¯p$
	means $1-p$.  And the definition extends to the entire $\BMS$ via linearity:
	\[Z \Bigl( ∑_j α_j \BSC(p_j) \Bigr)
	  ≔ ∑_j α_j Z(\BSC(p_j)) = ∑_j α_j √{p_j¯p_j}.\]
	This quantity can be seen as the expectation of the following random
	variable:
	\begin{itemize}
		\item select $\BSC(p_j)$ with probability $α_j$, and
		\item reveal $Z(\BSC(p_j))$, which is $2√{p_j¯p_j}$.
	\end{itemize}
	As an example, the Bhattacharyya parameter of $\BEC(ε) = ¯{ε} \BSC(0) + ε
	\BSC(1/2)$ is $¯{ε} Z(\BSC(0)) + ε Z(\BSC(1/2)) = ¯{ε}·0 + ε·1 = ε$.  The
	corresponding random variable follows the Bernoulli distribution with mean
	$ε$.

\subsection{Channel synthesis}

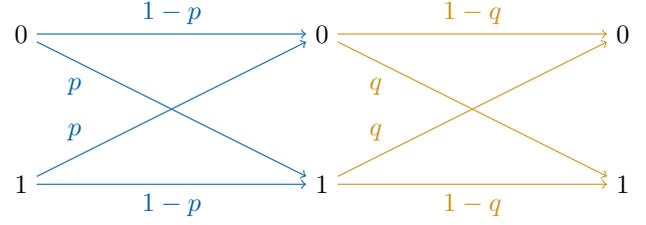
\begin{figure}\centering
	\begin{tikzpicture}
		\draw[xscale=4]
			(0,2)node(0){$0$}(1,2)node(a){$0$}(2,2)node(x){$0$}
			(0,0)node(1){$1$}(1,0)node(b){$1$}(2,0)node(y){$1$}
		;
		\draw[PMS3015]
			(0)edge[->]node[auto]{$1-p$}(a)
			(0)edge[->]node[auto,',pos=0.2]{$p$}(b)
			(1)edge[->]node[auto,pos=0.2]{$p$}(a)
			(1)edge[->]node[auto,']{$1-p$}(b)
		;
		\draw[PMS1245]
			(a)edge[->]node[auto]{$1-q$}(x)
			(a)edge[->]node[auto,',pos=0.2]{$q$}(y)
			(b)edge[->]node[auto,pos=0.2]{$q$}(x)
			(b)edge[->]node[auto,']{$1-q$}(y)
		;
	\end{tikzpicture}	
	\caption{
		Illustration of ${\color{PMS3015}\BSC(p)} Ｓ {\color{PMS1245}\BSC(q)}$,
		the serial combination of two BSCs.  A $0$ ends up flipped to $1$ with
		probability $p(1-q) + (1-p)q = p⋆q$.  Cf.\ \cite[Fig.~1.2]{LH06}.
	}\label{fig:BSC--BSC}
\end{figure}

\begin{figure}\centering
	\begin{tikzpicture}
		\draw[xscale=2]
							 (1,3)node(a){$0$}(3,3)node(w){$0$}(4,3)node(6){$00$}
			(0,2)node(0){$0$}(1,2)node(b){$1$}(3,2)node(x){$1$}(4,2)node(7){$11$}
			(0,1)node(1){$1$}(1,1)node(c){$0$}(3,1)node(y){$0$}(4,1)node(8){$01$}
							 (1,0)node(d){$1$}(3,0)node(z){$1$}(4,0)node(9){$10$}
		;
		\draw
			(0.east)edge[->]node[auto,sloped]{dup}(a)
			(0.east)edge[->](c)
			(1.east)edge[->](b)
			(1.east)edge[->]node[auto,',sloped]{dup}(d)
		;
		\draw[PMS3015]
			(a)edge[->]node[auto]{$1-p$}(w)
			(a)edge[->]node[auto,pos=0.8]{$p$}(x)
			(b)edge[->]node[auto,pos=0.2]{$p$}(w)
			(b)edge[->]node[auto,']{$1-p$}(x)
		;
		\draw[PMS1245]
			(c)edge[->]node[auto]{$1-q$}(y)
			(c)edge[->]node[auto,pos=0.8]{$q$}(z)
			(d)edge[->]node[auto,pos=0.2]{$q$}(y)
			(d)edge[->]node[auto,']{$1-q$}(z)
		;
		\draw
			(w)edge(6.west)(y)edge(6.west)
			(x)edge(7.west)(z)edge(7.west)
			(w)edge(8.west)(z)edge(8.west)
			(x)edge(9.west)(y)edge(9.west)
			(z)edge[opacity=0]node[opacity=1]{join}(9)
		;
	\end{tikzpicture}	
	\caption{
		Illustration of ${\color{PMS3015}\BSC(p)} Ｐ {\color{PMS1245}\BSC(q)}$,
		the parallel combination of two BSCs.  The output is conflictive ($01$
		or $10$) with probability $p⋆q$.  The output is consistent ($00$ or
		$11$) with probability $¬p⋆q¬$.  Cf.\ \cite[Fig.~1.3]{LH06}.
	}\label{fig:BSC||BSC}
\end{figure}
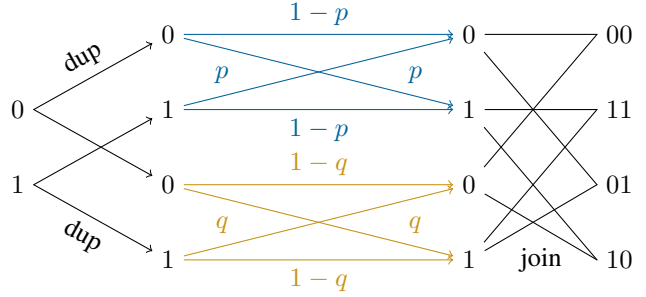

	We now define serial combinations and parallel combinations.  Readers are
	referred to \cite[Chapter~4]{RU08}, \cite{Arikan09}, and \cite{GR20} for
	more details.

	The \emph{serial combination} of two BMS channels $V$ and $W$ is denoted by
	$V Ｓ W$.  It is first defined for BSCs: $\BSC(p) Ｓ \BSC(q) ≔ \BSC(p⋆q)$,
	where $p⋆q ≔ p¯q + ¯pq$.  This new crossover probability satisfies $(¯p-p)
	(¯q-q) = ¬p⋆q¬ - p⋆q$, where $¬p⋆q¬ ≔ 1 - p⋆q = pq + ¯p¯q$.  See
	\cref{fig:BSC--BSC} for a picture.  Now extend the definition of serial
	combination to the whole $\BMS$ via bi-linearity:
	\begin{align*}
		～～&！！
		\Bigl( ∑_j α_j \BSC(p) \Bigr) Ｓ \Bigl( ∑_k β_k \BSC(q_k) \Bigr) \\
		& ≔∑_{jk} α_j β_k \BSC(p_j) Ｓ \BSC(q_k) \\
		& =∑_{jk} α_j β_k \BSC(p_j⋆q_k).
	\end{align*}
	When the two operands are equal, $W Ｓ W$ is also denoted by $W^ｓ$.

	The \emph{parallel combination} of two BMS channels $V$ and $W$ is denoted
	by $V Ｐ W$.  It is first defined for BSCs: $\BSC(p) Ｐ \BSC(q) ≔ p⋆q
	\BSC(÷{p¯q}{p⋆q}) + ¬p⋆q¬ \BSC(÷{pq}{¬p⋆q¬})$.  And then the definition is
	extended to the whole $\BMS$ via bi-linearity:
	\begin{align*}
		～～&！！
		\Bigl( ∑_j α_j \BSC(p) \Bigr) Ｐ \Bigl( ∑_k β_k \BSC(q_k) \Bigr) \\
		& ≔ ∑_{jk} α_j β_k \BSC(p_j) Ｐ \BSC(q_k) \\
		& = ∑_{jk} α_j β_k (p_j⋆q_k) \BSC \Bigl( ÷{p_j¯q_k}{p_j⋆q_k} \Bigr) \\
		&～ +∑_{jk} α_j β_k (¬p_j⋆q_k¬) \BSC \Bigl( ÷{p_jq_k}{¬p_j⋆q_k¬} \Bigr).
	\end{align*}
	When the two operands are equal, $W Ｐ W$ is also denoted by $W^ｐ$.

\subsection{Bhattacharyya equality}

	Bhattacharyya parameter is a special parameter in that parallel combination
	of channels translates to multiplication of $Z$'s.

	\begin{thm}\label{thm:parall-goto}
		For any BMS channel $W$,
		\[Z(VＰW) = Z(V) Z(W).\]
		In particular, $Z(W^ｐ) = Z(W)²$.
	\end{thm}

	\begin{proof}
		We first show that equality holds for $V$ and $W$ being BSCs.  Assume $V
		= \BSC(p)$ and $W = \BSC(q)$.  Then $V Ｐ W = p⋆q \BSC(÷{p¯q}{p⋆q}) +
		¬p⋆q¬ \BSC(÷{pq}{¬p⋆q¬})$.  The two component BSCs have Bhattacharyya
		parameters
		\[Z \Bigl( \BSC \Bigl( ÷{p¯q}{p⋆q} \Bigr) \Bigr)
		  = 2 √{ ÷{p¯q}{p⋆q} ¬\Bigl( ÷{p¯q}{p⋆q} \Bigr)¬ }
		  = ÷ { 2 √{p¯q¯pq}} {p⋆q} \]
		and
		\[Z \Bigl( \BSC \Bigl( ÷{pq}{¬p⋆q¬} \Bigr) \Bigr)
		  = 2 √{ ÷{pq}{¬p⋆q¬} ¬\Bigl( ÷{pq}{¬p⋆q¬} \Bigr)¬ }
		  = ÷ {2√{pq¯p¯q}} {¬p⋆q¬}.\]
		Overall, $\BSC(÷{p¯q}{p⋆q})$ is with weight $p⋆q$ so it contributes
		$2√{p¯pq¯q}$ to the Bhattacharyya parameter; $\BSC(÷{pq}{¬p⋆q¬})$ is
		with weight $¬p⋆q¬$ so it also contributes $2√{p¯pq¯q}$ to the
		Bhattacharyya parameter.  In sum, $Z(V Ｐ W) = 4√{p¯pq¯q} = Z(V) Z(W)$.

		The rest follows from the linearity of $Z$ and the bi-linearity of $Ｐ$
		in the BSC-decomposition.  More precisely, let $V$ and $W$ have
		BSC-decompositions $V = ∑_j α_j V_j$ and $W = ∑_k β_k W_k$, where
		$V_j$ and $W_k$ are BSCs.  Then $VＰW$ has BSC-decomposition $∑_{jk}
		α_j β_k V_j Ｐ W_k$ and Bhattacharyya parameter
		\begin{align*}
			∑_{jk} α_j β_k Z(V_j Ｐ W_k)
			& =∑_{jk} α_j β_k Z(V_j) Z(W_k) \\
			& =∑_j α_j Z(V_j) ∑_k β_k Z(W_k) \\
			& =Z(V) Z(W).
		\end{align*}
		This finishes the proof.
	\end{proof}

\section{Old Proof of \texorpdfstring{$μ ≤ 4.714$}{μ ≤ 4.714}}\label{sec:old}

	This section follows \cite{MHU16} and gives a self-contained proof of $μ ≤
	4.174$.

\subsection{Bhattacharyya inequalities}

	This subsection follows \cite[Exercise 4.62 (iv)]{RU08} and proves an
	inequality concerning Bhattacharyya parameters.

	Define a function $f： [0,1]² → [0,1]$ by
	\[f(x,y) ≔ √{x² + y² - x²y²}.\]

	\begin{lem}
		For $0 ≤ p,q ≤ 1$ we have
		\[f \bigl( Z(\BSC(p)), Z(\BSC(q)) \bigr) = Z(\BSC(p) Ｓ \BSC(q)).\]
	\end{lem}

	\begin{proof}
		The left-hand side is
		\begin{align*}
			f( 2√{p¯p}, 2√{q¯q} )
			& =√{ 4p¯p + 4q¯q - 16p¯pq¯q } \\
			& =2√{ p¯p(q+¯q)² + (p+¯p)²q¯q - 4p¯pq¯q } \\
			& =2√{ (p¯q+¯pq) (pq+¯p¯q) } \\
			& =2√{ (p⋆q) (¬p⋆q¬) } \\
			& =Z(\BSC(p⋆q)), 
		\end{align*}
		which is equal to the right-hand side.
	\end{proof}

	A bi-variate function $f(x,y)$ is said to be \emph{bi-convex} if the
	function is convex in $x$ for any fixed $y$ and convex in $y$ for any fixed
	$x$.

	\begin{lem}
		$f(x,y)$ is bi-convex.
	\end{lem}
	\begin{proof}
		Take the second derivative of $f$ in $x$:
		\[÷{∂²f}{∂x²} (x,y) = ÷ {y²(1-y^2)} {f(x,y)³}.\]
		This fraction is well-defined and nonnegative when $0 < y ≤ 1$.  Along
		the $y = 0$ segment, $f$ evaluates to $√{x²}$ and this is convex in $x$.
		Therefore $f$ is convex in $x$ for any fixed $y$.  For convexity in the
		$y$-direction we invoke symmetry.  This finished the proof
	\end{proof}

	\begin{thm}\label{thm:serial-under}
		For $V,W ∈ \BMS$ we have
		\[Z(V Ｓ W) ≥ f(Z(V), Z(W)).\]
		Equality holds when $V$ and $W$ are BSCs.
	\end{thm}

	\begin{proof}
		Let $V$ and $W$ have BSC-decompositions $∑_j α_j V_j$ and $∑_k β_k
		W_k$, respectively, where $V_j$ and $W_k$ are BSCs.  Then $V Ｓ W$ has
		BSC-decomposition $∑_{jk} α_j β_k V_j Ｓ W_k$ and Bhattacharyya
		parameter
		\[∑_{jk} α_j β_k Z(V_j Ｓ W_k) = ∑_{jk} α_j β_k f(Z(V_j), Z(W_k)).\]
		Let $X$ be a random variable that takes value $Z(V_j)$ with probability
		$α_j$.  Let $Y$ be an independent random variable that takes value
		$Z(W_k)$ with probability $β_k$.  Now we want to show
		\[Z(V Ｓ W) = 𝔼f(X,Y) ≥ f(𝔼X,𝔼Y) = f(Z(V),Z(W)).\]
		The left-hand side is greater than or equal to $𝔼f(X,𝔼Y)$ because $f$
		is convex in $y$ for each $x = Z(V_j)$.  The right-hand side is less
		than or equal to $𝔼f(X,𝔼Y)$ because $f$ is convex in $x$ for a fixed
		$y=  𝔼Y$.  This finishes the proof.
	\end{proof}

	An interesting consequence of the preceding argument is that the upper bound
	on $Z(VＳW)$ follows consequently.

	\begin{cor}\label{cor:serial-over}
		For any $V,W∈\BMS$, we have
		\[Z(VＳW) ≤ Z \bigl( \BEC(Z(V)) Ｓ \BEC(Z(W)) \bigr).\]
		Equality holds when $V$ and $W$ are BECs.
	\end{cor}

	\begin{proof}
		Continue the notation from the previous proof.  Now we vary the random
		variables $X$ and $Y$ but fix their expectations.  Then $𝔼f(X,Y)$
		varies while $f(𝔼X,𝔼Y)$ remains unchanged.  By Karamata's inequality,
		a corollary of Jensen's inequality, $𝔼f(X,Y)$ becomes larger when $X$
		and $Y$ becomes more marjorized.  The most marjorized random variables
		taking values in $[0,1]$ are those that can only be $0$ or $1$.  Those
		correspond to the BSC-decompositions of BECs, which consist of
		$\BSC(0) = \BSC(1)$ (with Bhattacharyya parameter $0$) and $\BSC(1/2)$
		(with Bhattacharyya parameter $1$).  Therefore, $Z(VＳW)$ is maximized
		when $V$ and $W$ are BECs.  This finishes the proof.
	\end{proof}

	\begin{cor}
		For any BMS channel $W$ with $z = Z(W)$,
		\begin{align*}
			z√{2-z²} ≤ Z(W^ｓ) & ≤ 2z-z², \\
			           Z(W^ｐ) & = z².
		\end{align*}
	\end{cor}

\subsection{Eigenfunction and eigenvalue}

\begin{figure}
	\pgfmathdeclarefunction{h_quadratic}1{%
		\ifdim#1pt=0pt%
			\def\pgfmathresult{0}%
		\else\ifdim#1pt=100pt%
			\def\pgfmathresult{0}%
		\else%
			\pgfmathsetmacro\a{#1*(10-#1)}%
			\pgfmathparse{\a^.78*((#1)^2+150)/500}%
		\fi\fi
	}
	\begin{tikzpicture}\centering
		\draw plot[domain=0:10,samples=500](\x*.8,{h_quadratic(\x)});
	\end{tikzpicture}
	\caption{
		$x^{0.78} (1-x)^{0.78} (2x²+3)$, a smooth instance of eigenfunction that
		induces supremum of ratios $0.87$ and overestimate $μ < 5$.
	}\label{fig:h-quadratic}
\end{figure}
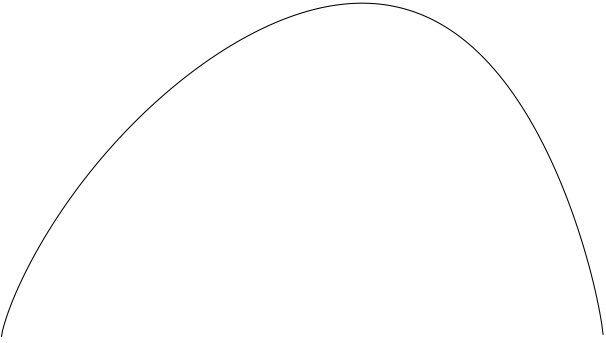

	Let $h： [0,1] → ℝ$ be a concave function such that $h(0) = h(1) = 0$ but
	positive elsewhere.  An overestimate of the scaling exponent can be obtained
	via the following relation
	\[λ ≔ \sup_{W∈\BMS*} ÷{h(Z(W^ｐ)) + h(Z(W^ｓ))} {2h(Z(W))}
	    ≥ 2^{-1/μ}. \label{sup:h(Z)}\]
	Recall that $\BMS*$ is the collection of all equivalence classes of BMS
	channels where $0 < Z(W) < 1$.

	To see why the quotient governs the scaling behavior, note that the
	``eigenvalue'' $λ$ is accumulative when we consider $W$'s children,
	grandchildren, grand-grandchildren, and so on.  To be more precise, we have
	\[h(Z(W^ｓ)) + h(Z(W^ｐ)) ≤ 2λ h(Z(W))\]
	and
	\begin{align*}
		～～&！！
		h(Z(W^{ｓｓ})) + h(Z(W^{ｓｐ}) + h(Z(W^{ｐｓ})) + h(Z(W^{ｐｐ})) \\
		& ≤ 2 λ h(Z(W^ｓ)) + 2 λ h(Z(W^ｐ)) \\
		& ≤ 4 λ h(Z(W)).
	\end{align*}
	And it is not hard to imagine
	\begin{align*}
		～～&！！
		h(Z(W^{ｓｓｓ})) + \dotsb + h(Z(W^{ｐｐｐ})) \\
		& ≤ 2 λ h(Z(W^{ｓｓ}))+\dotsb + 2 λ h(Z(W^{ｐｐ})) \\
		& ≤ 4 λ² h(Z(W^ｓ)) + 4 λ² h(Z(W^ｐ)) \\
		& ≤ 8 λ³ h(Z(W)).
	\end{align*}
	In general, when we consider all descendants $W ^ {?₁ ?₂ \dotso ?_n}$ at the
	$n$th generation, the average of $h(Z(W ^ {?₁ ?₂ \dotso ?_n}))$ cannot
	exceed $λ^n h(Z(W))$.  This quantity is exponentially small.  This implies
	that the $Z$ of deep enough descendants are generally very close to $0$ or
	to $1$, hence the polarization phenomenon.

	In our proof of $μ ≤ 4.63$, we will use eigenvalue to infer the scaling
	exponent without elaborating on the gap to capacity of an actual polar code.
	For the machinery that translates the eigenvalue into the asymptotic
	behavior of polar codes, see \cite{MHU16} or \cite[Sections
	2.4--2.6]{Chilly}.

	Since we know $Z(W^ｐ) = Z(W)²$ and we know how to bound $Z(W^ｓ)$ using
	functions in $Z(W)$, \cref{sup:h(Z)} assumes a simpler expression:
	\[\sup_{0<x<1\vphantom{√{x²}}}\; \sup_{x√{2-x²}≤y≤2x-x²}
	  ÷{h(x²)+h(y)}{2h(x)}.\label{sup:h(y)}\]
	As an example, $h(x) ≔ x^{0.78} (1-x)^{0.78} (2x²+3)$ leads to a supremum of
	$0.87$ and an upper bound of $μ ≤ 4.98$.  This eigenfunction is plotted in
	\cref{fig:h-quadratic}.

\subsection{Power iteration}

	To obtain a good function $h$ that minimizes \cref{sup:h(Z),sup:h(y)}---and
	thereby minimizing the overestimate of $μ$---consider the following
	inductive assignment:
	\begin{align*}
		h₀(x)      & ≔ x^{0.78} (1-x)^{0.78} (2x²+3), \\
		h_{n+1}(x) & ≔ \sup_{x√{2-x²}≤y≤2x-x²} ÷ {h_n(x²)+h_n(y)} {2\max h_n}.
	\end{align*}
	This is very similar to power iteration, an algorithm that approximates the
	longest eigenvalue of a square matrix.  For this reason $h$ is analogously
	called an \emph{eigenfunction} and quotients of the form $(h+h) / 2h$ are
	called \emph{eigenvalues}.

	It is unlikely that $h_n$ has a simple algebraic formula for large $n$.
	To proceed, one puts several ticks on $[0,1]$ 
	\[L ≔ \Bigl\{ 0, ÷1{ℓ}, …, ÷{ℓ-1}{ℓ}, 1 \Bigr\}\]
	and let $H ∈ ℝ^{ℓ+1}$ be an array parametrized by $L$.  The idea is to use
	$\LI (L,H)$ as a substitute of $h$ both during power iteration and when we
	want to overestimate $μ$.

	So we let a computer execute the following program.
	\[\left|～\begin{tabular}{l}
		For all $x∈L$: \\
		～～ $H[x] ← x^{0.78} (1-x)^{0.78} (2x²+3)$; \\
		Loop until $H$ converges: \\
		～～ $h ← \LI (L,H)$; \\
		～～ For all $x∈L$: \\
		～～～～ $H'[x] ← \dfrac {h(x²)+h(y(H,x))} {2h(x)\max H}$; \\
		～～ $H ← H'$;
	\end{tabular}\right.\]
	Here,
	\begin{itemize}
		\item $H'$ is an auxiliary array that holds the new content of $H$;
		\item $h： [0,1] → ℝ$ is a function such that $h(x) = H[x]$ for $x ∈ L$
		      and linearly interpolated for $x ∉ L$;
		\item $y(H,x)$ is the argument $y$ that maximizes $h(y)$
			  over the range $x√{2-x^2} ≤ y ≤ 2x-x²$.
	\end{itemize}
	We remark that there is an easy, i.e., $O(1)$, implementation of $y(H,x)$:
	\[y(H,x)≔\begin{cases*}
		x√{2-x²}  & if $x√{2-x²} ≥ \argmax H$, \\
		2x - x²   & if $2x-x² ≤ \argmax H$, \\
		\argmax H & otherwise.
	\end{cases*}\label{for:arg-may}\]
	This implementation is sound if $h$ is unimodal.  This might not be the case
	halfway the power iteration; but it deals no damage as long as $H$ converges
	and induces a good bound.

	Empirically, $H$ converges fast.  About $200$ iterations is enough to make
	$H$ and $H'$ differ by $10^{-15}$.  As a comparison, IEEE 754's
	double-precision floating-point format has $53$ significant bits (including
	the implicit leading $1$) and a relative precision of $2.22·10^{-16}$.

	Now that $H$ converges, let $ˆH$ be the limit of $H$ and let $ˆh$ be
	$\LI(L,ˆH)$.  An empirical upper bound of $μ$ is obtained by
	\[\biggl( -㏒₂ \max_{x∈L\setminus\{0,1\}}
	          ÷ {ˆh(x²)+ˆh(y(ˆH,x))} {2h(x)} \biggr) ^{-1}.\label{for:log2(H)}\]
	Per our computation, $ℓ = 2·10⁵ $ gives the first four digits ($4.695$)
	mentioned in \cite{MHU16} (wherein $ℓ = 10⁶$).

	We also tested using a variant of Chebyshev nodes as $L$:
	\[L ≔ \Bigl\{ ÷{1-\cos(θ)}2
	      \Bigm| θ = 0, ÷1{ℓ}π, …, ÷{ℓ-1}{ℓ}π, π \Bigr\}. \label{for:cheby}\]
	The motivation behind Chebyshev nodes is that they pay more attentions to
	the two ends of the interval, the places where $h(x)$ becomes small and more
	precisions are needed.  We found that $ℓ=2·10³$ gives the first four digits
	($4.695$), which indicates that Chebyshev nodes is superior than evenly
	spaced ticks.

\subsection{Foot of the mountain}

	Having an array $ˆH$ of evaluations, one would ask if $ˆh ≔ \LI(L,H)$ is a
	proper substitute of the eigenfunction in the manner of whether
	\[μ ≤ \biggl( -㏒₂ \max_{0<x<1} ÷{ˆh(x²)+ˆh(y(ˆH,x))}{2h(x)} \biggr)^{-1}\]
	gives a finite upper bound.  Unfortunately, no.  When $x$ is in $[0, 1/2ℓ]$
	or in $[1-1/2ℓ, 1]$, the interpolant is locally linear and the quotient
	$(ˆh(x²) + ˆh(2x-x²)) / 2ˆh(x)$ is constantly $1$ (whereas we want it to be
	strictly less than $1$).

	In \cite[Section~III.C]{MHU16}, it is explained how to manipulate $ˆh$ to
	obtain a proper eigenfunction that gives a more rigorous bound on the
	eigenvalue.  The strategy is to let $δ$ be a tiny number; and let $ˆh(x)$ be
	$x^{0.78}$ when $x≤δ$ and be $(1-x)^{0.78}$ when $x≥1-δ$.  This way, the
	quotients for $0<x<δ$ and for $1-δ<x<1$ are uniformly bounded from above.
	For $δ≤x≤1-δ$, since the denominator $2h(x)$ is far away from $0$, rounding
	error and sampling error can be controlled if we evaluate the quotient at a
	sufficiently fine set of points.

	This type of function surgery is limited to very tiny neighborhoods $[0, δ]$
	and $[1-δ, 1]$ of $0$ and $1$, respectively.  Hence it shall not affect the
	eigenvalue too much.  As an example, the empirical estimate obtained by
	\cref{for:log2(H)} is $4.695$; and the rigorous value reported in
	\cite{MHU16} is $μ ≤ 4.714$.  These two numbers are only $0.4\%$ apart.

	For our new overestimate of $μ$, we will skip the surgery step and use
	\cref{for:log2(H)}, the maximum over a discrete but very fine lattice, as
	an upper bound on the scaling exponent.

\subsection{Road map to a better bound}
	
	While taking \cref{sup:h(Z),sup:h(y)}, $y$ ranges over an interval
	$[x√{2-x²},2x-x²]$ where the left endpoint is tight if $W$ is a BSC and the
	right endpoint is tight if $W$ is a BEC.  If $W$ is a BEC, then all
	descendants of $W$ are BECs and $2x-x²$ is always tight.
	
	On the contrary, if $W$ is a BSC, the left endpoint is only tight for now.
	After one parallel combination, $W^ｐ$ will no longer be a BSC, and $x
	√{2-x²}$ will not be tight anymore.  That is to say, there is always a tiny
	gap between $Z(W^{ｐｓ})$ and $Z(W^ｐ) √{2 - Z(W^ｐ)²}$.  If we can come up
	with a better lower bound than $x √{2-x²}$, then \cref{sup:h(y)} will be
	taken over a smaller region, which makes it smaller.

	The next section finds the better bound.

\section{Tri-variate Channel Transformation}\label{sec:tri}

	Consider the channel combination $(U Ｐ V) Ｓ W$.  See
	\cref{fig:BSC||BSC--BSC} for a visualization.  Define a function $g： [0,1]³
	→ [0,1]$ that satisfies
	\[g(Z(U), Z(V), Z(W)) = Z((U Ｐ V) Ｓ W)\]
	for all $U,V,W$ that are BSCs.  We can write $g$ more explicitly with the
	help of the following lemmas.

\subsection{Tri-variate Bhattacharyya function}

\begin{figure*}\centering
	\begin{tikzpicture}
		\draw[xscale=2]
							 (1,3)node(a){$0$}(3,3)node(w){$0$}(4,3)node(6){$00$}
			(0,2)node(0){$0$}(1,2)node(b){$1$}(3,2)node(x){$1$}(4,2)node(7){$11$}
			(0,1)node(1){$1$}(1,1)node(c){$0$}(3,1)node(y){$0$}(4,1)node(8){$01$}
							 (1,0)node(d){$1$}(3,0)node(z){$1$}(4,0)node(9){$10$}
			(6,3)node(A){00}
			(6,2)node(B){11}
			(6,1)node(C){01}
			(6,0)node(D){10}
		;
		\draw
			(0.east)edge[->]node[auto,sloped]{dup}(a)
			(0.east)edge[->](c)
			(1.east)edge[->](b)
			(1.east)edge[->]node[auto,',sloped]{dup}(d)
		;
		\draw[PMS3015]
			(a)edge[->]node[auto]{$1-p$}(w)
			(a)edge[->]node[auto,pos=0.8]{$p$}(x)
			(b)edge[->]node[auto,pos=0.2]{$p$}(w)
			(b)edge[->]node[auto,']{$1-p$}(x)
		;
		\draw[PMS1245]
			(c)edge[->]node[auto]{$1-q$}(y)
			(c)edge[->]node[auto,pos=0.8]{$q$}(z)
			(d)edge[->]node[auto,pos=0.2]{$q$}(y)
			(d)edge[->]node[auto,']{$1-q$}(z)
		;
		\draw
			(w)edge(6.west)(y)edge(6.west)
			(x)edge(7.west)(z)edge(7.west)
			(w)edge(8.west)(z)edge(8.west)
			(x)edge(9.west)(y)edge(9.west)
			(z)edge[opacity=0]node[opacity=1]{join}(9)
		;
		\draw[PMS144]
			(6)edge[->]node[auto]{$1-r$}(A)
			(6)edge[->]node[auto,pos=0.8]{$r$}(B)
			(7)edge[->]node[auto,pos=0.2]{$r$}(A)
			(7)edge[->]node[auto,']{$1-r$}(B)
			(8)edge[->]node[auto]{$1-r$}(C)
			(8)edge[->]node[auto,pos=0.8]{$r$}(D)
			(9)edge[->]node[auto,pos=0.2]{$r$}(C)
			(9)edge[->]node[auto,']{$1-r$}(D)
		;
	\end{tikzpicture}	
	\caption{
		Illustration of $({\color{PMS3015}\BSC(p)} Ｐ {\color{PMS1245}\BSC(q)})
		Ｓ {\color{PMS144}\BSC(r)}$, The output is conflictive ($01$ or $10$)
		with probability $p⋆q$.  The output is consistent ($00$ or $11$) with
		probability $¬p⋆q¬$.
	}\label{fig:BSC||BSC--BSC}
\end{figure*}
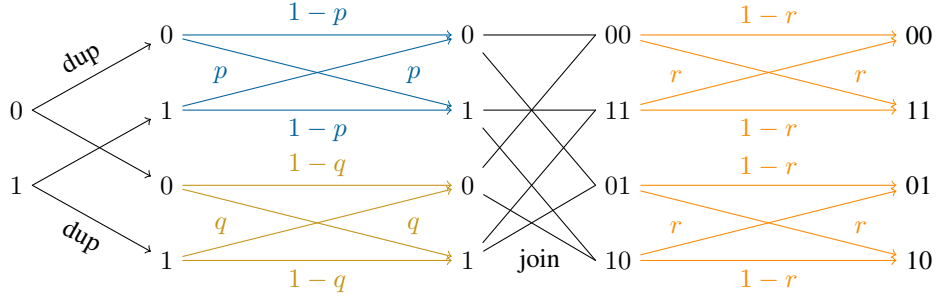

	\begin{lem}[Trivariate $Z$]\label{lem:tri-z}
		$(\BSC(p) Ｐ \BSC(q)) Ｓ \BSC(r)$ has Bhattacharyya parameter
		\[2√{ (p¯q¯r+¯pqr) (p¯qr+¯pq¯r) } + 2√{ (pq¯r+¯p¯qr) (pqr+¯p¯q¯r) }.\]
	\end{lem}

	\begin{proof}
		$\BSC(p) Ｐ \BSC(q)$ is, by definition, $p⋆q \BSC(÷{p¯q}{p⋆q}) + ¬p⋆q¬
		\BSC(÷{pq}{¬p⋆q¬})$.  When this channel is serially-combined with a
		$\BSC(r)$, the first summand becomes
		\[p⋆q \BSC \Bigl( ÷{p¯q}{p⋆q}⋆r \Bigr)
		  =p⋆q \BSC \Bigl( ÷{p¯q¯r+¯pqr}{p⋆q} \Bigr)\]
		and contributes Bhattacharyya parameter
		\[令÷\tfrac
			2 (p⋆q) √{ ÷{p¯q¯r+¯pqr}{p⋆q} ¬\bigl(÷{p¯q¯r+¯pqr}{p⋆q}\bigr)¬ }
			= 2 √{ (p¯q¯r+¯pqr) (p¯qr+¯pq¯r) }.\]
		The second summand becomes
		\[¬p⋆q¬ \BSC \Bigl( ÷{pq}{¬p⋆q¬}⋆r \Bigr)
		  = ¬p⋆q¬ \BSC \Bigl( ÷{pq¯r+¯p¯qr}{¬p⋆q¬} \Bigr)\]
		and contributes Bhattacharyya parameter
		\[令÷\tfrac
			2 ¬p⋆q¬ √{ ÷{pq¯r+¯p¯qr}{¬p⋆q¬} ¬\bigl(÷{pq¯r+¯p¯qr}{¬p⋆q¬}\bigr)¬ }
			= 2√{ (pq¯r+¯p¯qr) (pqr+¯p¯q¯r) }.\]
		This finishes the proof.
	\end{proof}

	\begin{lem}[$Z$ in terms of $Z$'s]\label{lem:z-in-z}
		\[g(x,y,z) = √{C+D} + √{C-D} = √{2C + √{C²-D²}},\]
		where
		\begin{align*}
			C & ≔ ÷14 (x²y² + ¬x²¬z² + ¬y²¬z²), \\
			D & ≔ ÷12 √{¬x²¬} · √{¬\smash y²¬}·z².
 		\end{align*}
	\end{lem}

	\begin{proof}
		Let $x$, $y$, and $z$ be $2√{p¯p}$, $2√{q¯q}$, and $2√{r¯r}$,
		respectively, for some $0 ≤ p,q,r ≤ 1/2$.  From \cref{lem:tri-z},
		$g(x,y,z)$ is $2√A + 2√B$, where
		\begin{align*}
			A & ≔ (p¯q¯r+¯pqr) (p¯qr+¯pq¯r) \\
			  & = p¯q¯rp¯qr + p¯q¯r¯pq¯r + ¯pqrp¯qr + ¯pqr¯pq¯r \\
			  & = p²¯q²r¯r + p¯pq¯q¯r² + p¯pq¯qr² + ¯p²q²r¯r, \\
			  & = p¯pq¯q(r²+¯r²) + (p²¯q²+¯p²q²)r¯r \\
			\shortintertext{and}
			B & ≔ (pq¯r+¯p¯qr) (pqr+¯p¯q¯r) \\
			  & = pq¯rpqr + pq¯r¯p¯q¯r + ¯p¯qrpqr + ¯p¯qr¯p¯q¯r \\
			  & = p²q²r¯r + p¯pq¯q¯r² + p¯pq¯qr² + ¯p²¯q²r¯r \\
			  & = p¯pq¯q(r²+¯r²) + (p²q²+¯p²¯q²)r¯r.
		\end{align*}
		To show $C+D = 4A$ and $C-D = 4B$, it suffices to show $2(A+B) = C$ and
		$2(A-B) = D$.  For the former,
		\begin{align*}
			2(A+B)
			& = 2\left(\begin{array}{l}
			    p¯pq¯q(r²+¯r²) + (p²¯q²+¯p²q²)r¯r \\
			    {} + p¯pq¯q(r²+¯r²) + (p²q²+¯p²¯q²)r¯r
			  \end{array}\right) \\
			& = 4p¯pq¯q(r²+¯r²) + 2(p²+¯p²)(q²+¯q²)r¯r \\
			& = ÷14 x²y² \Bigl(1-÷{z²}2\Bigr)
			  + ÷12 \Bigl( 1-÷{x²}2 \Bigr) \Bigl( 1-÷{y²}2 \Bigr) z² \\
			& = C.
		\end{align*}
		The third equality makes use of the rewriting rules $4r¯r=z²$ and
		$r² + ¯r² = (r+¯r)² - 2r¯r = 1 - z²/2$.  For the latter,
		\begin{align*}
			2(A-B)
			& = 2\left(\begin{array}l
			     p¯pq¯q(r²+¯r²)+(p²q²+¯p²¯q²)r¯r \\
			     {}-p¯pq¯q(r²+¯r²)-(p²¯q²+¯p²q²)r¯r
			    \end{array}\right) \\
			& = 2 (¯p²-p²) (¯q²-q²) r¯r \\
			& = 2 (¯p-p) (¯q-q) r¯r \\
			& = ÷12 √{1-x²} · √{1-y²}·z² \\
			& = D.
		\end{align*}
		The fourth equality makes use of the rewriting rule $(¯p-p)² = (¯p+p)² -
		4¯pp = 1-x²$.  In conclusion, we have $√{4A} + √{4B} = √{C+D} + √{C-D} =
		√{ \bigl( √{C+D} + √{C-D} \bigr)² } = √{ 2C + 2√{C²-D²} }$.  This
		finishes the proof.
	\end{proof}

	A tri-variate function is said to be \emph{tri-convexity} if it is convex
	whenever any two arguments are fixed and the other argument is varying.  If
	$g$ happens to be tri-convex, we will be able to show that $Z((U Ｐ V) Ｓ W)$
	is lower bounded by $g (Z(U), Z(V), Z(W))$ by the same Jensen-argument as in
	\cref{thm:serial-under}.  Unfortunately, $g$ is not tri-convex.  The next
	subsection will find a workaround to this.

\subsection{Lower tri-convex envelope}

	$g$ as defined above is not convex in any of the three variables.  We thus
	attempt to find a lower bound of $g$ that is tri-convex so that Jensen's
	inequality applies.  Consider a function $˘g： [0,1]³ → [0,1]$ that reads
	\[˘g(x,y,z) ≔ \sup \{ θ(x,y,z) ｜ θ ≤ g † and is tri-convex† \},\]
	where the supremum runs over all functions $θ： [0,1]³ → [0,1]$ that are
	tri-convex and pointwise bound $g$ from below.  This is very similar to the
	definition of the lower convex envelope, the difference being that $θ$ is
	not convex but tri-convex.  (An example is that $x y z$ is tri-convex but
	not convex.)  We will refer to $˘g$ as the \emph{envelope} of $g$.

	\begin{thm}[Counterpart of \cref{thm:serial-under}]\label{thm:serial-novel}
		For $U,V,W ∈ \BMS$,
		\[Z((U Ｐ V) Ｓ W) ≥ ˘g \bigl( Z(U), Z(V), Z(W) \bigr).\]
		In particular, if $W = V^ｐ$ for some $V ∈ \BMS$,
		\[Z(W^ｓ) ≥ ˘g \bigl( √{Z(W)}, √{Z(W)}, Z(W) \bigr).\]
	\end{thm}

	\begin{proof}
		For the former, it suffices to prove $Z((U Ｐ V) Ｓ W) ≥ θ \bigl(
		Z(U),Z(V),Z(W) \bigr)$ for all tri-convex $θ$ that is also $≤ g$
		pointwise.  Fix a $θ$.  When $U,V,W$ are BSCs, the inequality we want to
		prove holds:
		\begin{align*}
			Z((U Ｐ V) Ｓ W)
			& = g(Z(U), Z(V), Z(W)) \\
			& ≥ θ(Z(U), Z(V), Z(W)).
		\end{align*}
		Now consider BSC-decompositions $U = ∑_i α_i u_i$ and $V = ∑_j β_j
		V_j$ and $W = ∑_k γ_k W_k$, where $U_i, V_j, W_k$ are BSCs.  Then $(U
		Ｐ V) Ｓ W$ becomes $∑_{ijk} α_i β_j γ_k (U_i Ｐ V_j) Ｓ W_k$, thereby
		having Bhattacharyya parameter
		\begin{align*}
			Z((U Ｐ V) Ｓ W) 
			& =∑_{ijk}α_i β_j γ_k Z((U_i Ｐ V_j) Ｓ W_k))	 \\
			& =∑_{ijk}α_i β_j γ_k g(Z(U_i), Z(V_j), Z(W_k)) \\
			& ≥∑_{ijk}α_i β_j γ_k θ(Z(U_i), Z(V_j), Z(W_k)) \\
			& ≥∑_{ij}α_i β_j θ (Z(U_i), Z(V_j), Z(W)) \\
			& ≥∑_i α_i θ(Z(U_i), Z(V), Z(W)) \\
			& ≥ θ (Z(U), Z(V), Z(W)).
		\end{align*}
		This finishes the proof of the lower bound on $Z((U Ｐ V) Ｓ W)$.  For
		the lower bound on $Z(W^ｓ)$, plug in $U = V$ and $W = V^ｐ$, and use
		the fact that $Z(W) = Z(V)²$.
	\end{proof}

\subsection{Approximate the envelop \texorpdfstring{$˘g$}{˘g}}

	Computing the envelop $˘g$ algebraically does not seem plausible nor
	possible.  Our approach is to approximate $˘g$ numerically over a mesh
	\[M ≔ \Bigl\{ 0, ÷1n, …, ÷{n-1}n, 1 \Bigr\}³ ⊆ [0,1]³.\]
	Here, $n$ is the resolution; say $n = 200$.  We next evaluate $g$ at this
	mesh and run a program that iteratively lowers any evaluation that breaks
	tri-convexity.

	In detail, let $G ∈ ℝ ^ {(n+1) × (n+1) × (n+1)}$ be an $(n+1) × (n+1) ×
	(n+1)$ array indexed by $M$.  Initialize $G$ as
	\[G[x,y,z] ← g(x,y,z)\]
	for all $(x,y,z) ∈ M$.  We call $G$ the data points.  If the following
	does not hold for some $(x,y,z) ∈ M$ and $x ≠ 0,1$:
	\[2G[x,y,z] ≤ G \Bigl[ x-÷1n, y, z \Bigr]
	            + G \Bigl[ x+÷1n, y, z \Bigr], \label{cri:convex}\]
	we say that the data point at $(x,y,z)$ is \emph{breaking the convexity
	along the $x$-direction}.  To correct that, we update this data point as
	follows
	\[G[x,y,z] ← ÷12 G \Bigl[ x-÷1n, y, z \Bigr]
	           + ÷12 G \Bigl[ x+÷1n, y, z \Bigr].  \label{for:descend}\]
	We also demand the convexity in $y$-direction and $z$-direction:
	\begin{align*}
		2G[x,y,z] ≤ G \Bigl[ x,y-÷1n,z \Bigr]
		          + G \Bigl[ x,y+÷1n,z \Bigr], \label{cri:convey}\\
		2G[x,y,z] ≤ G \Bigl[ x,y,z-÷1n \Bigr]
		          + G \Bigl[ x,y,z+÷1n \Bigr] \,\phantom. \label{cri:convez}
	\end{align*}
	If not, we update $G [x,y,z]$ similarly.

	We synthesize a program that constantly searches for instances of data
	points that break the convexity in any of the three directions and keeps
	lowering data points.  Below is the program; let us call it $\TC$:
	\[\left|～\begin{tabular}{l}
		For all $(x,y,z)∈M$: \\
		～～ $G[x] ← g(x,y,z)$; \\
		Loop until $G$ converges: \\
		～～ For all $(x,y,z) ∈ M$: \\
		～～～～ If \cref{cri:convex} fails: \\
		～～～～～～ Update via \cref{for:descend} \\
		～～～～ If \cref{cri:convey} fails: \\
		～～～～～～ Update similarly; \\
		～～～～ If \cref{cri:convez} fails: \\
		～～～～～～ Update similarly; \\
	\end{tabular}\right.\]

	It will stop when all three criteria are met modulo rounding error.
	Empirically, $G$ converges; mathematically, we can also prove that $G$
	converges.

	\begin{pro}\label{pro:converge}
		$\TC$ makes $G$ converge.  For any mesh point $(x,y,z) ∈ M$, the data
		point $G[x,y,z]$ converges to
		\[ˇG[x,y,z] ≔ \sup\{ Θ[x,y,z] ｜ Θ ≤ G † and is tri-convex† \}.\]
		The supremum is over all arrays $Θ ∈ ℝ^{(n+1) × (n+1) × (n+1)}$ that
		satisfy the discrete convexity criteria
		\cref{cri:convex,cri:convey,cri:convez} and $Θ ≤ G$ entry-wise.
	\end{pro}

	\begin{proof}
		$Θ ≡ 0$ is a lower bound on $G$; it remains to be a lower bound after an
		update of data point.  Thus $G$ keeps decreasing but stays nonnegative.
		By the monotone convergence theorem, $G$ converges.  Let $ˇG$ be the
		limit of $G$ after any order of updates.  It must be tri-convex because
		any data point that violates convexity should have been updated.

		Now notice that any tri-convex lower bound $Θ ≤ G$ remains to be a lower
		bound on $G$ after an update of $G$.  So any such $Θ$ maintains to be a
		lower bound on $ˇG$.  This means that $ˇG$ is greater than or equal to
		the supremum of all such $Θ$'s.  But $ˇG$ is itself a tri-convex lower
		bound of $G$ so $ˇG$ is equal to the supremum; the supremum is a
		maximum.
	\end{proof}

	Hereafter, $ˇG$ denotes both the empirical end result of $\TC$ and the
	supremum defined in \cref{pro:converge}.  We call $ˇG$ the \emph{discrete
	envelop} in contrast to the ``continuous'' envelop $˘g$.

	\begin{lem}\label{lem:tri-convex}
		$\LI (M, ˇG)$ is tri-convex if the data points $ˇG$ satisfy the discrete
		convexity \cref{cri:convex,cri:convey,cri:convez}.
	\end{lem}

	Here, $\LI(M,ˇG)： [0,1]³ → ℝ$ is a function that evaluates to $ˇG [x,y,z]$
	at $(x,y,z) ∈ M$, and is tri-linearly interpolated if $(x,y,z) ∉ M$.  A
	defining feature of multi-linear interpolation is that it is piecewise
	linear in any cardinal direction.

	\begin{proof}[Proof of the lemma]
		\arraycolsep1pt
		We shall prove this for a two dimensional $2 × 3$ grid; the general
		statement follows by a generalization of this argument.

		Let there be six numbers on a grid
		\[\begin{matrix}
			a & -\!- & b & -\!- & c \\
			| &      & | &      & | \\
			d & -\!- & e & -\!- & f
		\end{matrix}\label{gri:abcdef}\]
		such that $a + c ≥ 2b$ and $d + f ≥ 2e$, i.e., the data points are
		convex.  Let $ˇg$ be obtained by bi-linear interpolation such that
		\[\begin{matrix}
			ˇg(-1,1) & -\!- & ˇg(0,1) & -\!- & ˇg(1,1) \\
			   |     &      &    |    &      &    |    \\
			ˇg(-1,0) & -\!- & ˇg(0,0) & -\!- & ˇg(1,0)
		\end{matrix}\]
		corresponds to \cref{gri:abcdef}.

		We claim that $ˇg$ is convex at $(0,0)$ in the $x$-direction, that is,
		$ˇg(-ε,0) + ˇg(ε,0) ≥ 2g(0,0)$ for $0 ≤ ε ≤ 1$.  This is because
		\[ˇg(-ε,0) + ˇg(ε,0) = εd + ¯{ε}e+εf + ¯{ε}e ≥ 2e.\]
		Similarly, $ˇg$ is convex at $(0,1)$ in the $x$ direction, that is,
		$ˇg(-ε,1) + ˇg(ε,1) ≥ 2g(0,1)$.

		Now we claim that $ˇg$ is convex at $(0,y)$, where $0≤y≤1$, in the
		$x$-direction.  That is to say, $ˇg(-ξ,y) + ˇg(ξ,y) ≥ 2g(0,y)$ for $0 ≤
		ξ ≤ 1$.  This is because
		\begin{align*}
			～～&！！
			ˇg(-ξ,y) + ˇg(ξ,y) \\
			& = yˇg(-ξ,1) + ¯yˇg(-ξ,0) + yˇg(ξ,1) + ¯yˇg(ξ,0) \\
			& ≥ 2yˇg(0,1) + 2¯yˇg(0,0) = 2ˇg(0,y).
		\end{align*}
		
		This shows that the convexity on the boundary of the interpolation
		cells follows from the convexity of the data points.  For convexity
		within a cell it trivially holds because the value within a cell is
		defined through interpolation.  Hence the lemma is sound.
	\end{proof}

	We conclude that $\LI (M, ˇG)$, the tri-linear interpolant of the discrete
	envelop, can be used as a substitute of $˘g$, the continuous envelop.
	Together with \cref{thm:serial-novel}, we can now lower bound $Z(W^{ｐｓ})$
	with a concrete object $ˇG$ in place of the abstract object $˘g$.

	Bibliographical remark:  some of the arguments presented in this section
	share common elements with \cite{Witsenhausen74}.

	In the next section, we will demonstrate how to utilize this new lower bound
	in power iteration.

\section{Finite State Power Iteration}\label{sec:dfa}

	For this section, recall the lesson that finite state automata has some
	memory when digesting the input stream.  We develop a variant of power
	iteration that keeps track of whether a synthetic channel is obtained by
	serial or parallel combination.

\subsection{Finite state automata}

\pgfplotstableread[header=false]{
       0 1569627 3042360 4406575 5661353 6802312 7815989 8680041 9363685 9824363
   28284 1600168 3070720 4432753 5685332 6823907 7834838 8695610 9375259 9830928
   56568 1630640 3099033 4458878 5709250 6845433 7853611 8711094 9386735 9837380
   84852 1661067 3127299 4484961 5733125 6866913 7872327 8726509 9398123 9843722
  113135 1691453 3155521 4511003 5756955 6888346 7890986 8741854 9409423 9849950
  141418 1721775 3183696 4536991 5780724 6909711 7909569 8757115 9420624 9856057
  169700 1752082 3211833 4562952 5804471 6931050 7928115 8772320 9431749 9862055
  197980 1782330 3239922 4588858 5828157 6952320 7946584 8787439 9442773 9867936
  226260 1812530 3267970 4614723 5851801 6973543 7964996 8802488 9453709 9873697
  254538 1842686 3295975 4640546 5875404 6994719 7983350 8817466 9464554 9879337
  282814 1872787 3323933 4666315 5898944 7015824 8001628 8832357 9475299 9884850
  316831 1902902 3351853 4692062 5922464 7036907 8019867 8847193 9485964 9890259
  351713 1932964 3379725 4717754 5945921 7057920 8038028 8861941 9496527 9895557
  386507 1962972 3407555 4743403 5969336 7078883 8056131 8876616 9506999 9900734
  421126 1992928 3435342 4769008 5992709 7099797 8074176 8891219 9517378 9905790
  455306 2022833 3463081 4794558 6016018 7120641 8092143 8905734 9527655 9910719
  489375 2052692 3490782 4820082 6039306 7141458 8110069 8920190 9537850 9915533
  523038 2082502 3518436 4845550 6062531 7162204 8127917 8934558 9547941 9920224
  556461 2112272 3546047 4870978 6085713 7182901 8145706 8948852 9557938 9924791
  589694 2142003 3573616 4896365 6108853 7203550 8163435 8963072 9567840 9929232
  622586 2171685 3601137 4921697 6131930 7224127 8181085 8977203 9577639 9933542
  655638 2201330 3628624 4947002 6154988 7244677 8198695 8991274 9587351 9937734
  688424 2230926 3656062 4972250 6177982 7265155 8216226 9005254 9596958 9941798
  721171 2260482 3683457 4997458 6200932 7285583 8233695 9019159 9606469 9945732
  753874 2289998 3710808 5022626 6223838 7305961 8251102 9032988 9615883 9949537
  786307 2319466 3738110 5047738 6246679 7326268 8268430 9046727 9625192 9953207
  818778 2348903 3765376 5072822 6269499 7346545 8285714 9060404 9634411 9956751
  851016 2378291 3792592 5097850 6292253 7366751 8302919 9073989 9643523 9960160
  883143 2407636 3819768 5122838 6314964 7386905 8320061 9087497 9652536 9963432
  915182 2436938 3846902 5147786 6337632 7407009 8337141 9100927 9661450 9966565
  947018 2466194 3873987 5172676 6360234 7427040 8354140 9114265 9670256 9969561
  978886 2495425 3901035 5197547 6382815 7447043 8371094 9127539 9678970 9972440
 1010579 2524609 3928034 5222359 6405330 7466973 8387967 9140721 9687574 9975179
 1042192 2553746 3954990 5247129 6427801 7486851 8404777 9153824 9696077 9977780
 1073741 2582838 3981903 5271856 6450228 7506676 8421523 9166846 9704479 9980242
 1105137 2611883 4008767 5296523 6472589 7526427 8438188 9179776 9712769 9982560
 1136556 2640890 4035597 5321170 6494929 7546147 8454806 9192638 9720965 9984744
 1167840 2669851 4062377 5345757 6517201 7565794 8471342 9205407 9729050 9986821
 1199064 2698769 4089115 5370302 6539430 7585388 8487813 9218094 9737029 9988748
 1230240 2727644 4115810 5394806 6561614 7604928 8504219 9230701 9744903 9990527
 1261299 2756472 4142455 5419250 6583731 7624395 8520543 9243214 9752664 9992156
 1292482 2785259 4169074 5443674 6605828 7643829 8536818 9255655 9760329 9993634
 1323567 2813999 4195641 5468037 6627858 7663189 8553011 9268002 9767879 9994960
 1354556 2842696 4222163 5492360 6649842 7682495 8569138 9280266 9775324 9996134
 1385451 2871350 4248641 5516642 6671779 7701747 8585197 9292447 9782660 9997154
 1416252 2899957 4275067 5540864 6693648 7720924 8601173 9304532 9789882 9998020
 1447055 2928527 4301463 5565063 6715495 7740067 8617098 9316545 9797003 9998730
 1477773 2957050 4327807 5589202 6737273 7759135 8632940 9328461 9804009 9999284
 1508444 2985531 4354109 5613300 6759006 7778149 8648714 9340292 9810906 9999681
 1539073 3013969 4380368 5637357 6780693 7797107 8664420 9352037 9817694 9999920
}\tablenovel

\begin{figure}\centering
	\begin{tikzpicture}[domain=0:10,samples=500,blend mode=multiply]
		\draw[black,scale=.8]plot(\x,{\x*(20-\x)/10});
		\pgfplothandlerlineto
		\pgfplotstreamstart
		\def\x{0}
		\pgfplotstableforeachcolumn\tablenovel\as\col{
			\pgfplotstableforeachcolumnelement\col\of\tablenovel\as\y{
				\xdef\x{\the\numexpr\x+1}
				\pgfplotstreampoint{\pgfpoint{\x mm/5*0.8}{\y sp*1.49174}}
			}
		}
		\pgfplotstreampoint{\pgfpoint{8cm}{8cm}}
		\pgfplotstreamend
		\pgfsetstrokecolor{PMS3015}
		\pgfusepath{stroke}
		\draw[PMS1245,scale=.8]plot(\x,{\x*sqrt(200-\x^2)/10});
		\draw[black,scale=.8]plot(\x,{\x^2/10});
	\end{tikzpicture}
	\caption{
		From top-left to bottom-right: old upper bound of $2x - x²$, new lower
		bound of $ˇg (√x,√x,x)$, old lower bound of $x √{2-x²}$, and parallel
		combination's Bhattacharyya parameter $x²$.
	}\label{fig:novel}
\end{figure}
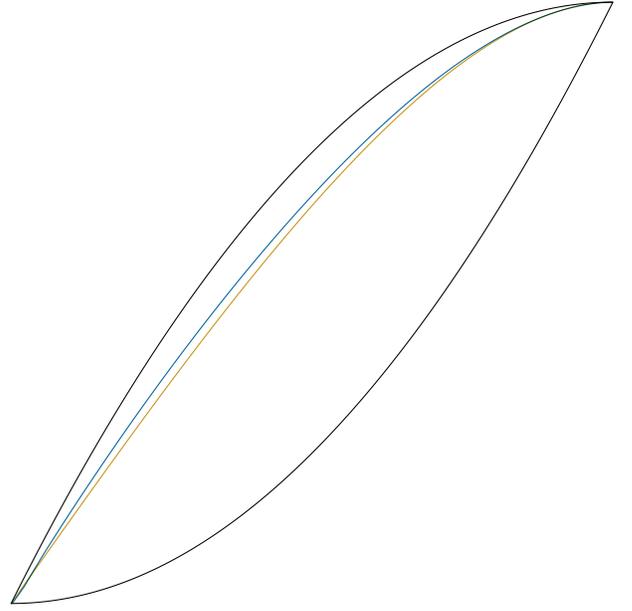

	To begin, suppose that there are two concave functions $φ_ｓ, φ_ｐ： [0,1] →
	ℝ$ that satisfy $φ_ｓ(0) = φ_ｓ(1) = φ_ｐ(0) = φ_ｐ(1) = 0$ but are positive
	elsewhere.  Define shorthands $ψ_ｓ, ψ_ｐ, ψ： \BMS → ℝ$ by
	\begin{align*}
		ψ_ｓ(W) & ≔ φ_ｓ(Z(W)), \\
		ψ_ｐ(W) & ≔ φ_ｐ(Z(W)), \\
		ψ(W)    & ≔ ψ_ｓ(W^ｓ) + ψ_ｐ(W^ｐ) \\
		        & ≔ φ_ｓ(Z(W^ｓ)) + φ_ｐ(Z(W^ｐ)).
	\end{align*}
	$ψ$ will be the counterpart of $h$ in our new bound.

	Here is the motivation of this indirect setup: in \cite{MHU16}, $h(Z(W))$ is
	a score that measures the extent of polarization---a smaller $h(Z(W))$ means
	that $W$ is more polarized.  Now we measure the extent of polarization of
	$W$ by first giving its children scores and sum them, except that we are
	biased.  As we will see later, $φ_ｓ(x)$ is greater than or equal to
	$φ_ｐ(x)$ for all $x$.  This means that, if $U^ｓ$ and $V^ｐ$ have the same
	Bhattacharyya parameter, we will give $V^ｐ$, a parallel combination, a
	lower score---because we think that $V^ｐ$ is more polarized.

	There is a reason to distinguish serial combination from parallel
	combination.  Comparing \cref{thm:serial-novel} with
	\cref{thm:serial-under}, we see that parallel combination assumes better
	bounds on Bhattacharyya parameters.  This implies that the domain of
	\cref{sup:h(y)} can be made smaller, which potentially makes the quotient
	corresponding to parallel combination smaller.

	Given the motivation, now we want a uniform upper bound on this ratio for
	all $W∈\BMS$:
	\begin{multline*}
		÷ {ψ(W^ｓ) + ψ(W^ｐ)} {2ψ(W)} \\
		= ÷ {ψ_ｓ(W^{ｓｓ}) + ψ_ｐ(W^{ｓｐ}) + ψ_ｓ(W^{ｐｓ}) + ψ_ｐ(W^{ｐｐ})}
		    {2ψ_ｓ(W^ｓ) + 2ψ_ｐ(W^ｐ)}.
	\end{multline*}
	Hence it suffices to bound
	\[÷ {ψ_ｓ(W^{ｓｓ}) + ψ_ｐ(W^{ｓｐ})} {2ψ_ｓ(W^ｓ)} ～†and†～
	  ÷ {ψ_ｓ(W^{ｐｓ}) + ψ_ｐ(W^{ｐｐ})} {2ψ_ｐ(W^ｐ)}\]
	from above.  One can now see the automata: channels that are serial
	combinations are always scored by $ψ_ｓ$, and channels that are parallel
	combinations are always scored by $ψ_ｐ$.  The subscript of $ψ$ indicates
	the current state of the automata; it remembers how the concerned channel
	was synthesized.

	We simplify the supremum of the first quotient as below:
	\begin{align*}
		～～&！！
		\sup_{W∈\BMS*} ÷ {ψ_ｓ(W^{ｓｓ}) + ψ_ｐ(W^{ｓｐ})} {2ψ_ｓ(W^ｓ)} \\
		& = \sup_{U=W^ｓ} ÷ {φ_ｓ(Z(U^ｓ)) + φ_ｐ(Z(U^ｐ))} {2φ_ｓ(Z(U))} \\
		& ≤ \sup_{U∈\BMS*} ÷ {φ_ｓ(Z(U^ｓ)) + φ_ｐ(Z(U^ｐ))} {2φ_ｓ(Z(U))} \\
		& = \sup_{0<x<1\vphantom f}\;
		    \sup_{f(x,x)≤y≤2x-x²} ÷ {φ_ｓ(x²)+φ_ｐ(y)} {2φ_ｓ(x)}.
	\end{align*}
	Here, the second supremum is taken over those $U$ that are themselves serial
	combinations.  We then treat $U$ as an usual BMS channel and apply the
	classic lower bound (\cref{thm:serial-under}). Because of that, $y$ ranges
	over $[f(x,x), 2x-x²]$.

	Similarly but not identically, the other quotient with $ψ_ｐ$ in the
	denominator can be simplified as below:
	\begin{align*}
		～～&！！
		\sup_{W∈\BMS*} ÷ {ψ_ｓ(W^{ｐｓ}) + ψ_ｐ(W^{ｐｐ})} {2ψ_ｐ(W^ｐ)} \\
		& =\sup_{V=W^ｐ} ÷ {φ_ｐ(Z(V^ｓ)) + φ_ｐ(Z(V^ｐ))} {2φ_ｐZ(V)} \\
		& ≤\sup_{0<x<1\vphantom{√x}}\;
		   \sup_{˘g(√x,√x,x)≤y≤2x-x²} ÷ {φ_ｐ(x²)+φ_ｓ(z)} {2φ_ｐ(x)}.
	\end{align*}
	Here, the second supremum is taken over those $V$ that are themselves
	parallel combinations.  We invoke \cref{thm:serial-novel} and let $z$ range
	over $[˘g (√x,√x,x), 2x-x²]$.  The new supremum is taken over a strictly
	smaller region than in the previous work---see \cref{fig:novel}---so a
	smaller supremum is expected.
	
\subsection{Power Iteration}

\pgfplotstableread[header=false]{
       0 3798380 6040041 7719288 8857789 9697779 9967306 9520300 8165274 5577032
  194759 3848966 6076949 7748010 8877924 9709158 9962617 9504274 8127634 5507259
  334017 3899412 6112684 7776597 8898163 9721123 9957734 9488315 8089169 5436834
  459349 3949889 6148300 7805079 8918313 9732815 9952416 9471305 8048975 5365523
  571410 4000085 6184028 7833583 8938552 9744999 9947244 9453184 8008015 5294851
  681069 4049697 6219993 7862010 8958884 9756506 9942203 9435243 7967408 5223246
  784124 4099682 6256496 7890227 8979202 9768163 9937475 9416494 7926414 5150533
  880187 4150385 6293308 7917823 8999592 9779580 9932830 9397248 7885056 5076813
  973417 4200521 6329992 7944816 9019654 9791378 9928697 9377980 7844248 5002498
 1066792 4250279 6366481 7970799 9039323 9803243 9924807 9357937 7802682 4925868
 1158178 4299793 6402871 7996473 9058988 9815136 9920546 9338062 7760657 4849139
 1245869 4348799 6438622 8021825 9078116 9826725 9916424 9317787 7719120 4772496
 1330834 4398190 6474721 8046081 9097469 9838515 9911615 9297471 7677878 4694685
 1412419 4447684 6510080 8070549 9116717 9849488 9906546 9276832 7635635 4614346
 1491005 4496658 6545187 8094768 9136063 9860486 9901244 9255817 7592040 4532480
 1568920 4546160 6580374 8118289 9155707 9870767 9895864 9234163 7547651 4449254
 1646505 4595976 6615679 8142011 9175192 9881078 9890267 9211813 7502561 4364329
 1724570 4643611 6651327 8165749 9193392 9890295 9884319 9189487 7457360 4276133
 1801535 4690514 6686520 8189822 9210290 9899279 9878390 9166678 7411837 4186025
 1877250 4737198 6721392 8213304 9226410 9908151 9871612 9143141 7365841 4094871
 1953227 4784516 6756034 8235659 9242104 9916430 9864570 9119113 7319229 4004560
 2026649 4831721 6791080 8257920 9258632 9923067 9857378 9094654 7270935 3913213
 2097425 4879573 6826144 8279205 9274174 9930184 9850273 9070473 7222009 3820792
 2168013 4927423 6861495 8300285 9289449 9936908 9841672 9045958 7172986 3727051
 2236504 4974880 6895645 8320856 9305086 9943400 9832457 9020622 7124624 3632043
 2303876 5022129 6928178 8341698 9321033 9949365 9823211 8994193 7076326 3534587
 2369278 5069068 6960337 8362636 9336008 9955282 9813446 8967108 7027208 3436184
 2433445 5115187 6992985 8383071 9351182 9961285 9803113 8938526 6976998 3337106
 2497348 5159935 7024971 8403281 9366643 9966962 9793598 8909175 6924775 3236223
 2560939 5204628 7057636 8423630 9382571 9972565 9784051 8878815 6872622 3133430
 2624939 5248071 7089802 8444286 9398924 9977995 9773602 8847772 6819590 3029385
 2688184 5291324 7122221 8465154 9415911 9982541 9762738 8817813 6764625 2923525
 2751221 5334425 7154470 8486021 9432773 9985891 9751628 8786062 6709031 2813861
 2814638 5377197 7187519 8506058 9449758 9989216 9740801 8753337 6652616 2700528
 2878415 5419988 7220943 8525998 9466147 9991842 9729974 8720015 6595427 2582144
 2940206 5462477 7254299 8546063 9482291 9993714 9719248 8686739 6536879 2463561
 3001489 5504775 7287264 8566354 9498230 9995409 9707887 8653320 6477181 2343524
 3062600 5545920 7320549 8586803 9514339 9996380 9696207 8619897 6416443 2221042
 3124563 5586399 7353371 8607385 9529741 9997668 9684493 8586184 6354284 2095072
 3185945 5625703 7386101 8628179 9545220 9998883 9672245 8552461 6291379 1966724
 3246579 5665053 7418392 8649314 9560156 9999880 9659725 8518628 6228634 1835328
 3305000 5703965 7451042 8670468 9575118 9999655 9646942 8484772 6166234 1699009
 3362237 5743141 7482046 8691521 9590482 9998584 9634344 8450918 6103127 1553675
 3418855 5782597 7512579 8712755 9604861 9996222 9621331 8417218 6039831 1405605
 3475425 5820519 7542907 8733947 9618480 9993291 9608252 8382748 5976218 1252081
 3531714 5857578 7573498 8755178 9631992 9989320 9594249 8348181 5911831 1092586
 3586641 5893977 7603212 8776189 9645433 9985572 9580504 8312186 5845668  921204
 3640335 5930448 7631878 8796709 9658698 9982144 9565697 8275660 5779577  739409
 3693315 5967143 7661052 8817147 9672912 9977524 9550727 8239188 5713624  541565
 3746778 6003466 7690216 8837724 9685975 9972427 9535583 8202434 5646201  316809
}\tableserial

\pgfplotstableread[header=false]{
       0 3798380 6040041 7719288 8857789 9635261 9838295 9392768 8076108 5550531
  194759 3848966 6076949 7748010 8877924 9646000 9834317 9376842 8039489 5481301
  334017 3899412 6112684 7776597 8898163 9657144 9830369 9360558 8002802 5411657
  459349 3949889 6148300 7805079 8918313 9667704 9825705 9343706 7964842 5341855
  571410 4000085 6184028 7833583 8938552 9678722 9821240 9325197 7925678 5272472
  681069 4049697 6219993 7862010 8958884 9689026 9816632 9306810 7886539 5202067
  784124 4099682 6256496 7890227 8979202 9699535 9811399 9287892 7847023 5130373
  880187 4150385 6293308 7917823 8999159 9709389 9806594 9268947 7806819 5057635
  973417 4200521 6329992 7944816 9018168 9719995 9802043 9250201 7767079 4984212
 1066792 4250279 6366481 7970799 9035904 9730751 9797631 9231685 7726998 4908763
 1158178 4299793 6402871 7996473 9052921 9741040 9793018 9213396 7686393 4833015
 1245869 4348799 6438622 8021825 9069517 9751270 9788380 9194621 7645798 4757406
 1330834 4398190 6474721 8046081 9086364 9761221 9783450 9175901 7605266 4680426
 1412419 4447684 6510080 8070549 9102236 9770189 9777689 9156405 7564142 4601106
 1491005 4496658 6545187 8094768 9118141 9779172 9771326 9136660 7521931 4519907
 1568920 4546160 6580374 8118289 9134523 9787721 9764939 9115863 7479357 4437489
 1646505 4595976 6615679 8142011 9150878 9795909 9758708 9093909 7435817 4353026
 1724570 4643611 6651327 8165749 9166216 9803097 9752023 9072253 7391779 4265564
 1801535 4690514 6686520 8189822 9181395 9810095 9745674 9050138 7347358 4176442
 1877250 4737198 6721392 8213304 9196590 9817089 9739591 9027957 7302760 4086025
 1953227 4784516 6756034 8235659 9212030 9823857 9733334 9005222 7257277 3996467
 2026649 4831721 6791080 8257920 9227921 9829696 9726943 8981341 7210000 3905791
 2097425 4879573 6826144 8279205 9244031 9835784 9720504 8957692 7162649 3814026
 2168013 4927423 6861495 8300285 9260038 9841766 9713322 8933260 7115325 3720953
 2236504 4974880 6895645 8320856 9276347 9847147 9705102 8908520 7068074 3626605
 2303876 5022129 6928178 8341698 9292247 9852192 9696182 8883301 7020522 3529653
 2369278 5069068 6960337 8362636 9307435 9857321 9686443 8856897 6972161 3431699
 2433445 5115187 6992985 8383071 9322601 9862161 9676238 8829946 6923037 3333162
 2497348 5159935 7024971 8403281 9337947 9866563 9666245 8802194 6871881 3232926
 2560939 5204628 7057636 8423630 9353120 9870651 9655985 8773121 6820101 3130670
 2624939 5248071 7089802 8444286 9368650 9874739 9645222 8743556 6767854 3027201
 2688184 5291324 7122221 8465154 9383850 9878204 9634347 8714197 6713776 2921747
 2751221 5334425 7154470 8486021 9398970 9880857 9623590 8683189 6659747 2812492
 2814638 5377197 7187519 8506058 9414186 9882969 9613105 8651455 6605457 2699507
 2878415 5419988 7220943 8525998 9428903 9883980 9601971 8619070 6549675 2581489
 2940206 5462477 7254299 8546063 9442648 9884381 9590896 8586333 6492684 2463271
 3001489 5504775 7287264 8566354 9456217 9883980 9579238 8553589 6433992 2343524
 3062600 5545920 7320549 8586803 9469786 9882807 9567685 8521249 6374779 2221042
 3124563 5586399 7353371 8607385 9482556 9881158 9556703 8489323 6313820 2095072
 3185945 5625703 7386101 8628179 9495914 9879705 9545732 8457194 6251830 1966724
 3246579 5665053 7418392 8649314 9509268 9878758 9534379 8424494 6190334 1835328
 3305000 5703965 7451042 8670468 9522259 9876322 9522674 8390994 6129463 1699009
 3362237 5743141 7482046 8691521 9535228 9872930 9510096 8357605 6067650 1553675
 3418855 5782597 7512579 8712755 9548063 9868840 9496614 8324147 6005477 1405605
 3475425 5820519 7542907 8733947 9560807 9864761 9483062 8290474 5943196 1252081
 3531714 5857578 7573498 8755178 9573692 9860202 9468630 8256805 5879832 1092586
 3586641 5893977 7603212 8776189 9586628 9855977 9453940 8221621 5815106  921204
 3640335 5930448 7631878 8796709 9599330 9851586 9439108 8185621 5750301  739409
 3693315 5967143 7661052 8817147 9612268 9847034 9423780 8149128 5685107  541565
 3746778 6003466 7690216 8837724 9624281 9842458 9408270 8112695 5618609  316809
}\tableparall

\begin{figure}\centering
	\begin{tikzpicture}[blend mode=multiply]
		\pgfplothandlerlineto
		\pgfplotstreamstart
		\def\x{0}
		\pgfplotstableforeachcolumn\tableserial\as\col{
			\pgfplotstableforeachcolumnelement\col\of\tableserial\as\y{
				\xdef\x{\the\numexpr\x+1}
				\pgfplotstreampoint{\pgfpoint{\x mm/5*0.8}{\y sp*2}}
			}
		}
		\pgfplotstreampoint{\pgfpoint{8cm}{0cm}}
		\pgfplotstreamend
		\pgfsetstrokecolor{PMS3015}
		\pgfusepath{stroke}
		\pgfplothandlerlineto
		\pgfplotstreamstart
		\def\x{0}
		\pgfplotstableforeachcolumn\tableparall\as\col{
			\pgfplotstableforeachcolumnelement\col\of\tableparall\as\y{
				\xdef\x{\the\numexpr\x+1}
				\pgfplotstreampoint{\pgfpoint{\x mm/5*0.8}{\y sp*2}}
			}
		}
		\pgfplotstreampoint{\pgfpoint{8cm}{0cm}}
		\pgfplotstreamend
		\pgfsetstrokecolor{PMS1245}
		\pgfusepath{stroke}
	\end{tikzpicture}	
	\caption{
		Eigenfunction pair $ˆ{φ}_ｓ$ (blue) and $ˆ{φ}_ｐ$ (brown).  The former is
		greater for $x > 0.4$---this is the place where $z(x, ˆ{Φ}_ｓ) > y(x,
		ˆ{Φ}_ｓ)$ due to $ˇg (√x,√x,x) > f(x,x)$.
	}\label{fig:eigen}
\end{figure}
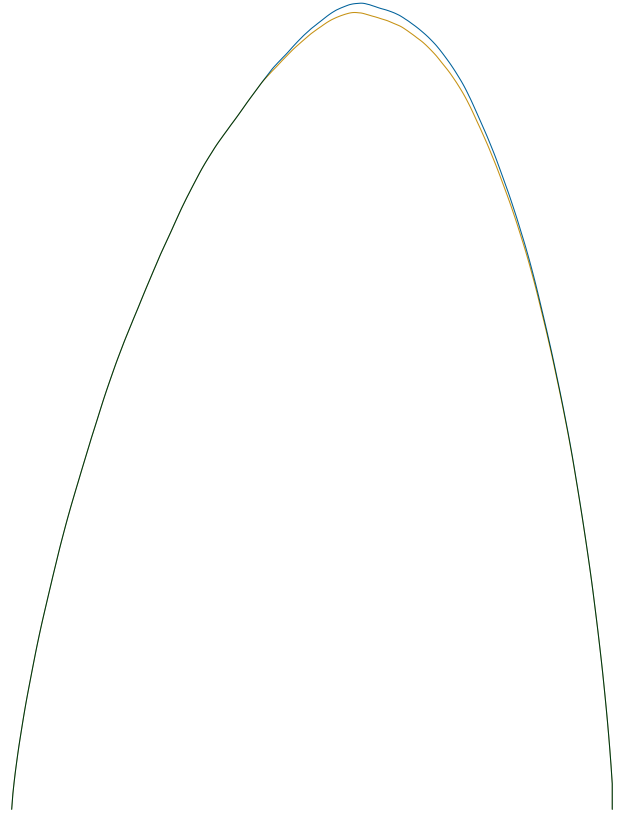

	It remains to use linear interpolation to represent $φ_ｓ$ and $φ_ｐ$, and
	apply power iteration to minimize the eigenvalues.

	Let $L$ be \cref{for:cheby}; say $ℓ = 10⁶$.  Let $Φ^ｓ, Φ^ｐ ∈ ℝ^{ℓ+1}$ be
	arrays parametrized by $L$.  We execute this program:
	\[\left|～\begin{tabular}{l}
		For all $x∈L$: \\
		～～ $Φ_ｓ[x] ← x^{0.78} (1-x)^{0.78} (2x²+3)$; \\
		～～ $Φ_ｐ[x] ← x^{0.78} (1-x)^{0.78} (2x²+3)$; \\
		Loop until $Φ_ｓ$ and $Φ_ｐ$ converge: \\
		～～ $φ_ｓ ← \LI(L,Φ_ｓ)$; \\
		～～ $φ_ｐ ← \LI(L,Φ_ｐ)$; \\
		～～ For all $x∈L$: \\
		～～～～ $Φ_ｓ'[x] ← \dfrac {φ_ｐ(x²)+φ_ｓ(y(Φ_ｓ,x))} {2φ_ｓ(x)\maxΦ_ｓ}$; \\
		～～～～ $Φ_ｐ'[x] ← \dfrac {φ_ｐ(x²)+φ_ｓ(z(Φ_ｓ,x))} {2φ_ｐ(x)\maxΦ_ｓ}$; \\
		～～ $Φ_ｓ ← Φ_ｓ'$; \\
		～～ $Φ_ｐ ← Φ_ｐ'$; \\
	\end{tabular}\right.\]
	Here,
	\begin{itemize}
		\item $Φ_ｓ'$ and $Φ_ｐ'$ are temporary memory spaces
		      that store the updated content for the next round.
		\item $y(Φ_ｓ,x)$ and $z(Φ_ｓ,x)$ are meant to be the arguments 
		      that maximize $φ_ｓ(y)$ and $φ_ｓ(z)$ over the ranges
		      $f(x,x) ≤ y ≤ 2x-x²$ and $ˇg (√x,√x,x) ≤ z ≤ 2x-x²$, respectively.
		\item $ˇg$ is $\LI (M, ˇG)$, which is $≈ ˘g$.
		      If a rigorous lower bound of $˘g$ is desired, see \cref{app:mpfi}.
	\end{itemize}
	We can reuse the implementation of $y(H,x)$ in \cref{for:arg-may};  and
	implement $z(H,x)$ as
	\[z(H,x) ≔ \begin{cases*}
		ˇg (√x,√x,x) & if $ˇg (√x,√x,x) ≥ \argmax H$, \\
		2x-x²        & if $2x-x² ≤ \argmax H$, \\
		\argmax H    & otherwise.
	\end{cases*}\]

	Empirically, $Φ_ｓ$ and $Φ_ｐ$ converge.  Let $ˆ{Φ}_ｓ$ and $ˆ{Φ}_ｐ$ be the
	end results of power iteration.  We can now use
	\begin{align*}
		ˆ{φ}_ｓ & ≔ \LI(L, ˆ{Φ}_ｓ) \\
		ˆ{φ}_ｐ & ≔ \LI(L, ˆ{Φ}_ｐ)
	\end{align*}
	as the scoring functions.  See \cref{fig:eigen} for their plots; notice that
	$ˆ{φ}_ｓ ≥ ˆ{φ}_ｐ$.

\section{New Proof of \texorpdfstring{$μ≤4.63$}{μ ≤ 4.63}}\label{sec:wrap}

	This section gathers the materials and proves the main theorem.

	\begin{thm}[Main theorem]
		$μ ≤ 4.63$, where $μ$ is the scaling exponent of polar coding using
		Arıkan's kernel $[¹₁{}⁰₁]$ over BMS channels.
	\end{thm}

	\begin{proof}
		We have seen that  $\LI (M, ˇG) ≈ ˘g$ and $˘g (√x,√x,x) ≤ Z(W^ｓ)$,
		where $W$ is a parallel combination of another BMS channel and $x =
		Z(W)$.  To obtain a practical yet rigorous lower bound on $Z(W^ｓ)$,
		see \cref{app:mpfi} for how to define $G_↘$ and $ˇG_↘$.  By
		\cref{thm:better} therein, we have $ˇg_↘ (√x,√x,x) ≤ Z(W^ｓ)$ where
		$ˇg_↘ ≔ \LI (M, ˇG_↘)$.

		Next, apply power iteration to optimize for the eigenvalues
		\begin{align*}
			λ_ｓ & ≔\sup_{x∈L∖\{0,1\}\vphantom f}\;
			        \sup_{f(x,x)≤z≤2x-x²}
			        ÷ {ˆ{φ}_ｓ(x²) + ˆ{φ}_ｐ(z)} {2ˆ{φ}_ｓ(x)}, \\
			λ_ｐ & ≔\sup_{x∈L∖\{0,1\}\vphantom{√x}}\;
			        \sup_{ˇg_↘(√x,√x,x)≤z≤2x-x²}
			        ÷ {ˆ{φ}_ｐ(x²) + ˆ{φ}_ｓ(z)} {2ˆ{φ}_ｐ(x)}.
		\end{align*}
		Per our execution, both suprema are about $0.860714$.

		Finally, we conclude that
		\[\sup_{W∈\BMS*} ÷{ψ(W^ｓ) + ψ(W^ｐ)} {2ψ(W)}\]
		has $\max(λ_ｓ, λ_ｐ) ≈ 0.860715$ as an empirical upper bound.  And $μ$
		has $㏒₂ (\max(λ_ｓ, λ_ｐ)) ^ {-1} ≈ 4.62125$ as an empirical upper bound.
		Hence it is safe to say $μ ≤ 4.63$.
	\end{proof}

\section{Conclusions}

	In this paper, we argue that the scaling exponent is an essential constant
	characterizing the scaling behavior of polar coding, of which very little is
	known.  We then lower the overestimate of the scaling exponent from $4.714$
	to $4.63$.
	
	The limit of this method---analyzing $(U Ｐ V) Ｓ W$ to gain better control
	on $Z$---is $4.61126$.  This number is obtained by assuming $g$ tri-convex
	and using $g (√x,√x,x) = x (1 + √{5-4x²}) / 2$ as the lower bound on the
	$Z(W^{ｐｓ})$ in terms of $x = Z(W^ｐ)$.  Futhermore, we expect that
	analyzing $(U Ｐ V) Ｓ (W Ｐ X)$ leads to a better bound.

\appendices
\crefalias{section}{appendix}

\section{Linear Interpolation Made a Proper Lower Bound}\label{app:mpfi}

\pgfmathdeclarefunction{example_f}1{%
	\pgfmathsetmacro\t{3.1416/1.618*(#1)}%
	\pgfmathparse{(\t-sin(\t r))/4}%
}

\begin{figure}\centering
	\begin{tikzpicture}[xscale=4/3]
		\draw[PMS3015]plot[domain=0:6,samples=200](\x,{example_f(\x)});
		\draw[PMS1245](0,0)foreach\x in{1,...,6}{--(\x,{example_f(\x)})};
	\end{tikzpicture}
	\caption{
		Piecewise linear interpolation (brown) of an arbitrary function (blue)
		is an approximation but not a valid lower bound.
	}\label{fig:linear}
\end{figure}
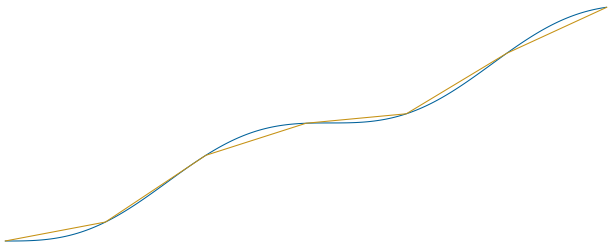

	There is a caveat when approximating $˘g$ using $ˇG$: the mesh is coarse.
	For one-dimensional interpolation (i.e., $H$ and $Φ_ｓ$ and $Φ_ｐ$), we can
	afford arrays of size $10⁶$ and the error is negligible as we only cares
	about the first three digits of the scaling exponent.  Unlike the
	one-dimensional case, for a three-dimensional mesh, the cube of $200$ is
	already $8·10⁶$ but the error is of the order of $1/200$.  See
	\cref{fig:linear} for an illustration of the caveat.
	
	In this appendix, we will demonstrate how to find an array $G_↘$ such that
	$\LI (M, G_↘) ≤ g$ pointwise.  With $G_↘$, we can run the iterative
	algorithm $\TC$ and the resulting array $ˇG_↘$ will satisfy $\LI (M, ˇG_↘)
	≤ ˘g$ pointwise.  This will give us a mathematically rigorous control on
	$Z(W^{ｐｓ})$.

\subsection{Monotonic increasing approach}

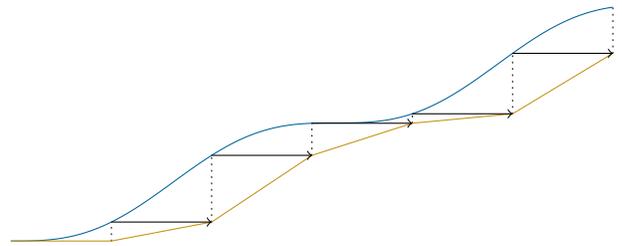
\begin{figure}\centering
	\begin{tikzpicture}[xscale=4/3]
		\draw[dotted]foreach\x in{1,...,6}{
			(\x,{example_f(\x-1)})--(\x,{example_f(\x)})
		};
		\draw[PMS3015]plot[domain=0:6,samples=200](\x,{example_f(\x)});
		\draw[PMS1245](0,0)foreach\x in{1,...,6}{--(\x,{example_f(\x-1)})};
		\foreach\x in{2,...,6}{\draw[->](\x-1,{example_f(\x-1)})--+(1,0);}
	\end{tikzpicture}
	\caption{
		If the target function is monotonically increasing, the evaluation at an
		interval's left end is a lower bound over the interval.  Thus, shifting
		the interpolant $δ$ units right makes it a lower bound, where $δ$ is
		the width of the intervals.
	}\label{fig:increase}
\end{figure}

	Observe that $g(x,y,z)$ is a monotonic increasing function in $x$, $y$, and
	$z$.  This is a consequence of $x$, $y$, $z$, and $g$ being the
	Bhattacharyya parameters of certain BSCs.  In particular, we know $g(a,b,c)
	≤ g(x,y,z)$ for all $(x,y,z) ∈ (a,b,c) + [0,1/n]³$.  Here, the right-hand
	side is the mesh cell whose lower-left-near corner is $(a,b,c)$ and
	upper-right-far corner is $(a+1/n, b+1/n, c+1/n)$.

	Inspired by the observation, we declare a new array $G_→ ∈ ℝ ^ {(n+1) ×
	(n+1) × (n+1)}$ that is parametrized by $M$ and populated by
	\[G_→ [a,b,c] ← g \Bigl( a-÷1n∨0, b-÷1n∨0, c-÷1n∨0 \Bigr).\]
	Here, $a-1/n∨0$ means $\max(a - 1/n, 0)$.  We call this the \emph{monotonic
	increasing approach} and illustrate it in \cref{fig:increase}.  The
	following lemma shows that linearly interpolating this array serves as a
	lower bound.

	\begin{lem}
		$\LI (M, G_→) ≤ g$ pointwise.
	\end{lem}

	\begin{proof}
		It suffices to check the inequality cell-by-cell.  Fix an $(a,b,c) ∈ M$;
		we shall prove the inequality on the cell $(a,b,c) + [0, 1/n]³$.  Now
		for any $(x,y,z)$ in this cell, $\LI (M, G_→) (x,y,z)$ is a convex
		combination of these eight numbers
		\[令÷\tfrac\begin{matrix}
			～～ G_→ [a+÷0n, b+÷1n, c+÷1n] ～ G_→ [a+÷1n, b+÷1n, c+÷1n], \\
			    G_→ [a+÷0n, b+÷0n, c+÷1n], ～ G_→ [a+÷1n, b+÷0n, c+÷1n], ～～ \\[1ex]
			～～ G_→ [a+÷0n, b+÷1n, c+÷0n], ～ G_→ [a+÷1n, b+÷1n, c+÷0n], \\
			    G_→ [a+÷0n, b+÷0n, c+÷0n], ～ G_→ [a+÷1n, b+÷0n, c+÷0n], ～～
		\end{matrix}\]
		By the definition of $G_→$, all eight numbers are less than or
		equal to $g(a,b,c)$, so $\LI (M,G_→) (x,y,z) ≤ g(a,b,c) ≤ g(x,y,z)$.
	\end{proof}

	If we apply the monotonic increasing approach to a $200 × 200 × 200$ mesh,
	we get $μ ≤ 4.66359$.  To go below $4.63$, we have to combine this with a
	second approach introduced in the next subsection.

\subsection{Smoothness approach}

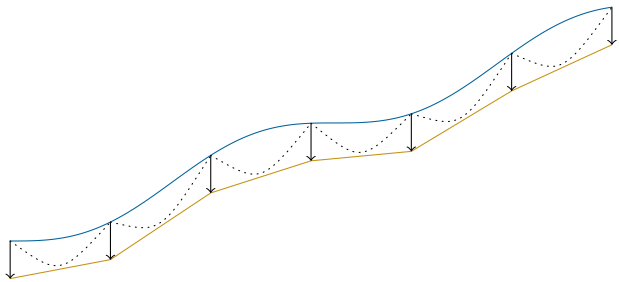
\begin{figure}\centering
	\begin{tikzpicture}[xscale=4/3]
		\foreach\x in{1,...,6}{
			\pgfmathsetmacro\y{example_f(\x-1)}
			\pgfmathsetmacro\z{example_f(\x)}
			\draw[dotted](\x-1,\y)..controls(\x-.5,\y/2+\z/2-0.6)..(\x,\z);
		}
		\draw[PMS3015]plot[domain=0:6,samples=200](\x,{example_f(\x)});
		\draw[PMS1245](0,-.5)foreach\x in{0,...,6}{--(\x,{example_f(\x)-.5})};
		\foreach\x in{0,...,6}{\draw[->](\x,{example_f(\x)})--+(0,-.5);}
		\path(-.1,0)(6.1,0);
	\end{tikzpicture}
	\caption{
		If the target function is smooth (the second derivative has an upper
		bound, $\sup f'' ≤ m$), it can be lower bounded by parabolas.  Thus,
		shifting the linear interpolant $mδ² / 8$ units down makes it a lower
		bound, where $δ$ is the width of the intervals.
	}\label{fig:smooth}
\end{figure}

	Idea: if we control two end points and the second derivative,
	we control the evaluations in between.

	\begin{lem}\label{lem:smooth-1d}
		Let $f： [0,1] → ℝ$ be doubly-differentiable on $[0,1]$.  Suppose $f(0)
		= f(1) = 0$ and $f''(x) ≤ m$ for some $m ≥ 0$.  Then for any $x ∈
		[0,1]$,
		\[f(x) ≥ -÷m8.\] 
	\end{lem}

	\begin{proof}
		As a special case of Lagrange interpolation, consider a linear
		interpolation using $(0, f(0))$ and $(0, f(1))$ as reference points.
		Its error term is $(0-x) (1-x) f''(y) / 2$ for some $y ∈ [0,1]$.
		Clearly $x (1-x) ≤ 1/4$ and this finishes the proof.
	\end{proof}

	\begin{lem}\label{lem:smooth-3d}
		Let $n$ be a positive integer.  Let $g： [0, 1/n]³ → ℝ$ be
		doubly-differentiable on $[0, 1/n]³$.  Suppose $g = 0$ at the eight
		corners of the cube $[0, 1/n]³$.  Suppose $g_{xx} ≤ m₁$ and $g_{yy} ≤
		m₂$ as well as $g_{zz} ≤ m₃$ for some $m₁,m₂,m₃ ≥ 0$.  Then for any
		$(x,y,z) ∈ [0, 1/n]³$,
		\[g(x,y,z) ≥ - ÷ {m₁+m₂+m₃} {8n²}.\]
	\end{lem}

	\begin{proof}
		First apply \cref{lem:smooth-1d} in the $x$-direction to lower bound
		$g(x,0,0)$, $g(x,0,1/n)$, $g(x,1/n,0)$, and $g(x,1/n,1/n)$ by $- m₁ /
		8n²$.  Then apply \cref{lem:smooth-1d} in the $y$-direction to lower
		bound $g(x,y,0)$ and $g(x,y,1/n)$ by $- (m₁+m₂) / 8n²$.  Finally, apply
		\cref{lem:smooth-1d} in the $z$-direction to lower bound $g(x,y,z)$ by
		$- (m₁+m₂+m₃) / 8n²$.
	\end{proof}

	\Cref{lem:smooth-3d} provides an excellent way to lower bound $g$ on a mesh
	as the denominator $8n²$ keeps up with the memory usage $O(n³)$ better than
	the monotonic increasing approach did, in which case the error was
	$O(g'/n)$,

	Let us declare a new array $G_↓ ∈ ℝ ^ {(n+1) × (n+1) × (n+1)}$ that is
	parametrized by $M$ and populated by
	\[G_↓ [a,b,c] ← g(a,b,c) - ÷ {m₁+m₂+m₃} {8n³},\]
	where
	\begin{align*}
		m₁ & = \sup_{((a,b,c) + [-1/n,1/n]³) ∩ [0,1]³} \max(g_{xx}, 0), \\
		m₂ & = \sup_{((a,b,c) + [-1/n,1/n]³) ∩ [0,1]³} \max(g_{yy}, 0), \\
		m₃ & = \sup_{((a,b,c) + [-1/n,1/n]³) ∩ [0,1]³} \max(g_{zz}, 0).
	\end{align*}
	The suprema are taken over all mesh cells that touch $(a,b,c)$.  The
	following lemma confirms that linearly interpolating $G_↓$ serves as a
	valid lower bound of $g$.  See also \cref{fig:smooth} for an illustration.

	\begin{lem}
		$\LI (M, G_↓) ≤ g$ pointwise.
	\end{lem}

	\begin{proof}
		It suffices to check the inequality cell-by-cell.  Fix an $(a,b,c) ∈ M$;
		we shall prove that the inequality holds on the cell $(a,b,c) + [0,
		1/n]³$.  At the eight corners of this cell, $g$ and $\LI (M, G)$
		coincide.  Hence $˜g ≔ g - \LI (M, G)$ is a function that is zero at the
		eight corners.  Its second derivatives $˜g_{xx}$, $˜g_{yy}$, and
		$˜g_{zz}$ are nothing but $g_{xx}$, $g_{yy}$, and $g_{zz}$,
		respectively.  Now apply \cref{lem:smooth-3d}: $˜g ≥ - (m₁+m₂+m₃) /
		8n³$, where $m₁,m₂,m₃$ are the suprema of the second derivatives over
		the concerned cell.  Hence
		\begin{align*}
			g
			& ≥ \LI (M, G) - ÷ {m₁+m₂+m₃} {8n³} \\
			& ≥ \LI (M, G_↓).
		\end{align*}
		This finishes the proof.
	\end{proof}

\subsection{Interval arithmetic for derivatives}

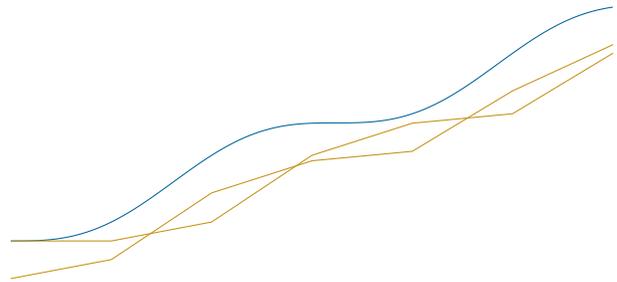
\begin{figure}\centering
	\begin{tikzpicture}[xscale=4/3]
		\draw[PMS3015]plot[domain=0:6,samples=200](\x,{example_f(\x)});
		\draw[PMS1245](0,-.5)foreach\x in{0,...,6}{--(\x,{example_f(\x)-.5})};
		\draw[PMS1245](0,0)foreach\x in{1,...,6}{--(\x,{example_f(\x-1)})};
	\end{tikzpicture}
	\caption{
		Both \cref{fig:increase,fig:smooth} are proper lower bounds.  Now we
		have the freedom to choose tighter bound on an interval-by-interval
		basis.
	}\label{fig:both}
\end{figure}

	In the previous subsection, we see how to initialize $G_↓$ in
	principle---evaluate $g$ at every mesh point and subtract by $1/8n²$ times
	the local suprema of second derivatives.  It remains to actually
	compute the second derivatives.
	
	The first shortcut we take is that $m₁,m₂,m₃$ do not have to be the exact
	suprema; any upper bounds serve the same purpose.  So it remains to bound
	the second derivatives from above for every cell.  In fact, since we have
	$8n²$ in the denominator, there is nearly no precision requirement; any
	$m₁,m₂,m₃$ that are $< 10$ will end up giving a better bound than $G_→$.

	The second shortcut we take is that there are softwares that can take care
	of differentiation.  Given the formula of $g$, SageMath, an open-source
	mathematical software system, computes its symbolic derivatives by passing
	the queries to Maxima, a classical open-source software that excels at
	algebra.

	Once the symbolic expressions of $g_{xx}$, $g_{yy}$, and $g_{zz}$ are
	obtained, the third---perhaps the biggest---shortcut we take is treating
	each cell as a fuzzy triple of real numbers and evaluating the expressions
	using interval arithmetic.  For example, the cell $(0.1,0.4,0.7) + [0,0.1]³$
	can be seen as an imperfect representation of three real numbers $x$, $y$,
	and $z$ that are approximately $0.15$, $0.45$, and $0.75$ with error radius
	$0.05$.  When evaluating, say, $xy - z$, all we know is that the true value
	must lie in the set
	\begin{multline*}
		\{ xy-z ｜ (x,y,z) ∈ (0.1,0.4,0.7) + [0,0.1]³ \} \\
		= [0.1·0.4-0.8, 0.2·0.5-0.7].
	\end{multline*}
	An interval arithmetic package takes cares of the tedious edge cases and
	returns an interval that \emph{provably} contains the true value of every
	mathematical expression.
	
	In our case, MPFI is the C-library SageMath calls behind the scene.  The
	abbreviation stands for multiple-precision floating-point interval.  A
	defining feature of the MPFI library is that it temporarily increases the
	precision during the evaluation process to narrow down the output interval.
	As an example, evaluating $x - x$ without simplification first will double
	the error radius.  But by cutting the interval into smaller pieces the
	result will be the union of smaller intervals surrounding $0$, hence
	improving the output precision.

\subsection{The better-of-the-two approach}

	Given two approaches, $G_→$ and $G_↓$, we see that $G_→$ is tighter at
	places where $g'$ is small but $g''$ is large; and $G_↓$ is tighter whenever
	$g'$ is big and $g''$ is far less than $8n²$.  In the sequel, we will let
	$G_↘$ be the array that uses values from $G_→$ or $G_↓$ depending on which
	is tighter.

	Consider a cell $(a,b,c) + [0, 1/n]³$ whose lower-left-near corner is at
	$(a,b,c)$ and upper-right-far corner is at $(a+1/n, \allowbreak b+1/n,
	\allowbreak c+1/n)$.  For every such cell, we want to decide whether to use
	the monotonic increasing approach or the smoothness approach.  We set a
	rule: we will use $G_→$ by default, but if $G_→$ is worse than $G_↓$ at all
	eight corners $(a,b,c) + \{0,1\}³$, we switch to $G_↓$.

	Now that we have specified which approach to use for every cell, we can
	initialize $G_↘$.  Intuitively speaking, $G_↘ [a,b,c]$ will be $G_→
	[a,b,c]$ if any cell that touches $(a,b,c)$ decides to go for the increasing
	approach, but will be $G_↓ [a,b,c]$ if all cells that touch $(a,b,c)$ decide
	to go for the smoothness bound.  A formal summary is as below,
	\begin{itemize}
		\item $G_↘ [a,b,c] = G_→ [a,b,c]$ iff for some mesh point $(x,y,z) ∈
		      (a,b,c) + \{-1/n, 0, 1/n\}³$ that shares a common cell with
		      $(a,b,c)$, the monotonic increasing approach is better: $G_→
		      [x,y,z] ≥ G_↓ [x,y,z]$.
		\item $G_↘ [a,b,c] = G_↓ [a,b,c]$ iff for all mesh points $(x,y,z) ∈
		      (a,b,c) + \{-1/n, 0, 1/n\}³$ that share a common cell with
		      $(a,b,c)$, the smoothness approach is better: $G_↓ [x,y,z] ≥ G_→
		      [x,y,z]$.
	\end{itemize}
	The following theorem concludes this appendix.

	\begin{thm}\label{thm:better}
		With $G_↘$ defined as above,  we have
		\[\LI (M, G_↘) ≤ g.\]
		With $ˇG_↘$ being the result of performing $\TC$ on $G_↘$,  we have
		\[ˇg_↘ ≔ \LI (M, ˇG_↘) ≤ ˘g.\]
		In particular, with $x ≔ Z(W)$ we have
		\[Z(W^ｓ) ≥ ˇg_↘ (√x,√x,x).\]
	\end{thm}

	\begin{proof}
		The first statement is by how $G_↘$ merges data points from $G_→$ and
		$G_↓$.  The second statement is by the first statement and
		\cref{lem:tri-convex}.  The last statement is by the second statement
		and \cref{thm:serial-novel}.
	\end{proof}

	For a faster way to convexify an array, see the next appendix.

\section{Convexify Faster}\label{app:scan}

	In this appendix, we describe a strategy to tri-convexify a
	three-dimensional array $G$.  This strategy converges faster than repeated
	uses of \cref{for:descend}.

	Consider a one dimensional array $A = \{a₀, …, a_n\}$ that is parametrized
	by $L = \{l₀, …, l_n\}$.  We want to lower some entries of $A$ so that
	$\LI(L,A)$ becomes convex.  This is equivalent to finding the convex hull of
	points \[(l₀,\max(A))\,,\, (l₀,a₀)\,,\, …\,,\, (l_n,a_n)\,,\,
	(l_n,\max(A)).\] We next apply Graham's scan.  Since the $l$-coordinates are
	already sorted, the time complexity of one scan is $O(n)$.  The output of
	Graham's scan is a list of points that support the convex hull.  For points
	that lie strictly inside, we update their $a$-values using linear
	interpolation.  This step also costs time complexity $O(n)$.

	Now that we know how to convexify one dimensional arrays, we iteratively
	apply this to the axes of the thee-dimensional array $G$.  Here, an axis of
	$G$ is data points where two coordinates are fixed and the other coordinate
	is varying.

	Since convexifying one axis only lowers the data points, $G$ is ever
	decreasing.  But since $G$ stays non-negative, it converges by monotone
	convergence theorem.

\def\bibsetup{
	\biblabelsep0pt
	\hbadness10000
	\interlinepenalty=1000\relax
	\widowpenalty=1000\relax
	\clubpenalty=1000\relax
	\frenchspacing
	\biburlsetup
}
\printbibliography

\end{document}